\documentclass{article}

\usepackage[T1]{fontenc}
\usepackage{repsize}
\usepackage{latexsym}
\usepackage{defs}
\usepackage{color}
\usepackage{amsmath}
\usepackage{amssymb}  
\usepackage{graphicx}

\newcommand{\emmedparagraph}[1]{{\medskip\noindent\emph{#1}}}
\newcommand{\etal}{{\em et al.}}

\newcommand{\mycase}[1]{\mbox{{\underline{Case #1}}:\/}}

\newcommand{\ignore}[1]{{}}

\makeatletter
\def\remark@name{\relax}
\def\named@remark[#1]#2{\def\remark@name{#1}\unnamed@remark{#2}}
\def\unnamed@remark#1{\marginpar[\hfill$\longrightarrow${\color{red}\sf\scriptsize#1}~{\color{blue}\remark@name}]{$\longleftarrow${\color{red}\sf\scriptsize#1~{\color{blue}\remark@name}}}}
\def\remark{\@ifnextchar[{\named@remark}{\unnamed@remark}}

\makeatother

\newcommand{\jstar}{{\ifmmode{j^\ast}\else{$j^\ast$}\fi}}
\newcommand{\tstar}{{\ifmmode{t^\ast}\else{$t^\ast$}\fi}}
\newcommand{\Rstar}{{\ifmmode{R^\ast}\else{$R^\ast$}\fi}}

\newcommand{\barB}{{\bar B}}

\newcommand{\calC}{{\cal C}}

\newcommand{\calG}{{\cal G}}

\newcommand{\calP}{{\cal P}}
\newcommand{\calS}{{\cal S}}

\newcommand{\calT}{{\cal T}}

\newcommand{\braced}[1]{{ \left\{ #1 \right\} }}

\newcommand{\floor}[1]{{\lfloor #1\rfloor}}

\newcommand{\NP}{{\mbox{\sf NP}}}
\newcommand{\APX}{{\mbox{\sf APX}}}

\newcommand{\MaxCC}{{\mbox{\sc MaxCC}}}
\newcommand{\MinCC}{{\mbox{\sc MinCC}}}

\newcommand{\singleton}{{\mbox{\it singleton}}}
\newcommand{\merge}{{\mbox{\it merge}}}

\newcommand{\Greedy}{{\sc Greedy}}
\newcommand{\OPT}{{\sf OPT}}
\newcommand{\smallOPT}{{\scriptstyle\sf OPT}}
\newcommand{\OCC}{{\sf OCC}}

\newcommand{\opt}{\mbox{\sf{O}}}
\newcommand{\alg}{\mbox{\sf{S}}}

\newcommand{\deltaalg}{\Delta}

\newcommand{\profit}{\mbox{\sf{profit}}}

\newcommand{\optprof}{\mbox{\sf{O}}}
\newcommand{\algprof}{\mbox{\sf{S}}}

\newcommand{\cost}{\mbox{\sf{cost}}}
\newcommand{\profvalue}[1]{{\binom{#1}{2}}}

\newcommand{\LL}{{\mbox{\tiny\rm L}}}
\newcommand{\RR}{{\mbox{\tiny\rm R}}}
\newcommand{\DD}{{\mbox{\tiny\rm D}}}
\newcommand{\rmL}{{\mbox{\rm L}}}
\newcommand{\rmR}{{\mbox{\rm R}}}

\newcommand{\KK}{\ensuremath{{\bf K}}}

\newcommand{\half}{{\textstyle\frac{1}{2}}}

\newcommand{\threehalves}{{\textstyle\frac{3}{2}}}
\newcommand{\onethird}{{\textstyle\frac{1}{3}}}
\newcommand{\twothirds}{{\textstyle\frac{2}{3}}}

\newcommand{\threefifths}{{\textstyle\frac{3}{5}}}
\newcommand{\onesixth}{{\textstyle\frac{1}{6}}}

\newcommand{\defeq}{\stackrel{\text{def}}{=}}

\newcommand{\GDY}{{\sf GDY}}

\newcommand{\R}{{\sf R}}

\begin{document}

\title{Online Clique Clustering\footnote{This paper combines and extends the results from conference
	publications~\cite{FabijanNP_clique_clustering_2013,Chrobak_etal_15}}}

\author{Marek Chrobak\thanks{Department of Computer Science and Engineering, University of California at Riverside, USA. Research supported by NSF grants CCF-1536026, 
CCF-0729071 and CCF-1217314. email:~{\tt marek@cs.ucr.edu} }
\and
Christoph D\"urr\thanks{Sorbonne Universit\'es, UPMC Univ Paris 06, UMR 7606, LIP6, Paris, France and CNRS, UMR 7606, LIP6, Paris, France. email:~{\tt cristoph.durr@lip6.fr} }
\and
Aleksander Fabijan\thanks{Department of Computer Science, Malm\"o University, SE-205~06~~Malm\"o, Sweden. email:~{\tt \{aleksander.fabijan, bengt.nilsson.TS\}@mah.se} }
\and
Bengt J.~Nilsson\footnotemark[4]
}


\maketitle

\begin{abstract}
Clique clustering is the problem of partitioning the vertices of
a graph into disjoint clusters, where each cluster forms a clique in the graph, while
optimizing some objective function. 
In online clustering, the input graph is given one vertex at a time, and any vertices that have previously been clustered together are not allowed to be separated. 
The goal is to maintain a clustering with an objective value close to the optimal solution.  For the variant where we want to maximize the number of edges in the clusters, we propose an online strategy based on the doubling technique.  It has an asymptotic competitive ratio at most 
$15.646$ and an absolute competitive ratio at most $22.641$. We also show that no deterministic strategy can have an asymptotic competitive ratio better than~$6$. 
For the variant where we want to minimize the number of edges between clusters, we show that the deterministic  competitive ratio of the problem is $n-\omega(1)$, where $n$ is the number of vertices in the graph.
\end{abstract}


\section{Introduction}
\label{sec: introduction}


The correlation clustering problem and its different variants have been extensively studied over the past decades; 
see e.g.~\cite{BansalBC04_correlation_2004,charikar2003clustering,DemaineImorlica03}. The instance of
correlation clustering consists of a graph whose vertices represent some objects and edges represent their
similarity.
Several objective functions are used in the literature, e.g., maximizing the number of edges within the clusters 
plus the number of non-edges between clusters (maximizing agreements), or minimizing the number of non-edges inside the clusters 
plus the number of edges outside them (minimizing disagreements). 
Bansal~{\etal}~\cite{BansalBC04_correlation_2004} show that both the minimization 
of disagreement edges and the maximization of agreement edges versions are 
$\NP$-hard. However, from the point of view of approximation the two versions differ. 
In the case of maximizing agreements this problem actually admits a PTAS, whereas in the case of minimizing 
disagreements it is $\APX$-hard. Several efficient constant factor approximation algorithms are proposed 
for minimizing 
disagreements~\cite{BansalBC04_correlation_2004,charikar2003clustering,DemaineImorlica03} and 
maximizing agreements~\cite{charikar2003clustering}. 

Some correlation clustering problems may impose additional restrictions on the structure
or size of the clusters. We study the variant, called \emph{clique clustering}, where
the clusters are required to form disjoint cliques in the underlying graph $G = (V,E)$.
Here, we can maximize the number of edges inside the clusters
 or minimize the number of edges outside the clusters. These measures give rise to the maximum and minimum clique clustering problems respectively. The computational complexity and approximability of these problems have attracted attention recently~\cite{Dessmark_etal_06,FigueroaGJKLP_approximate_clustering_2008,ShamirST_graph_modification_2004}, 
and they have numerous applications within the areas of gene expression profiling and DNA clone classification~\cite{Ben-DorSY_gene_expression_1999,FigueroaBJ_fingerprings_2004,ShamirST_graph_modification_2004,Valinsky_etal_bacterial_2002}. 

We focus on the online variant of clique clustering, where the input graph $G$ is not known in advance. (See~\cite{Borodin_ElYaniv_98} for more background on online problems.) The vertices of $G$ arrive one at a time. Let $v_t$ denote the vertex that arrives at time $t$, for $t = 1,2,\ldots$. When $v_t$ arrives, its edges to all preceding vertices $v_1,\ldots,v_{t-1}$ are revealed as well. In other words, after step $t$, the subgraph of $G$ induced by $v_1,v_2,\ldots,v_t$ is known, but no other information about $G$ is available. In fact, we assume that even the number $n$ of vertices
is not known upfront, and it is revealed only when the process terminates after step $t=n$.

Our objective is to construct a procedure that incrementally constructs and outputs a clustering based on the information acquired so far. Specifically, when $v_t$ arrives at step $t$, the procedure first creates a singleton clique $\braced{v_t}$. Then it is allowed to merge any number of cliques (possibly none) in its current partitioning into larger cliques. No other modifications of the clustering are allowed. The merge operation in this online setting is irreversible; once vertices are clustered together, they will remain so, and hence, a bad decision may
have significant impact on the final solution. This online model was proposed by Charikar~{\etal}~\cite{CharikarCFM_incremental_2004}.

We avoid using the word ``algorithm'' for our procedure, since it evokes connotations with computational limits in terms of complexity. In fact, we place no limits on the computational power of our procedure and, to emphasize this, we use the word \emph{strategy} rather than \emph{algorithm}.
This approach allows us to focus specifically on the limits posed by the lack of complete information about
the input. Similar settings have been studied in previous work on online computation, for example for online medians~\cite{Chrobak_Hurand_medians_11,Chrobak_KNY_bidding_08,Lin_NRW_incremental_10}, 
minimum-latency tours~\cite{Chaudhuri_GRT_tours_03}, and several other online optimization problems~\cite{Chrobak_Mathieu_doubling_2006},
where strategies with unlimited computational power were studied.


\paragraph{Our results.} We investigate the online clique clustering problem and provide upper and lower bounds for
the competitive ratios for its maximization and minimization versions, that we denote $\MaxCC$ and $\MinCC$, respectively. 

Section~\ref{sec: maxclustering} is devoted to the study of $\MaxCC$. We first observe that the competitive ratio of
the natural greedy strategy is linear in $n$. 
We then give a constant competitive strategy for $\MaxCC$ with asymptotic competitive ratio at most
$15.646$ and absolute competitive ratio at most $22.641$. The strategy is based on the doubling technique often used in online
algorithms. We show that the doubling approach cannot give a competitive ratio smaller than $10.927$.
We also give a general lower bound, proving
that there is no online strategy for $\MaxCC$ with competitive ratio smaller than~$6$.
Both these lower bounds apply also to asymptotic ratios.

In Section~\ref{sec: minclustering} we study online strategies for $\MinCC$.
We prove that no online strategy can have a competitive ratio of $n-\omega(1)$. We then
show that the competitive ratio of the greedy strategy is $n-2$, matching this
lower bound.


\section{Preliminaries}
\label{sec: preliminaries}


We begin with some notation and basic definitions of the $\MaxCC$ and $\MinCC$ clustering problems. 
They are defined on an input graph $G=(V,E)$, with vertex set $V$ and edge set~$E$.
We wish to find a partitioning of the vertices in $V$ into clusters so that each cluster induces a clique 
in  $G$. In addition, we want to optimize some objective function associated with the clustering. 
In the $\MaxCC$ case, this objective 
is to maximize the total number of edges inside the clusters, whereas in the $\MinCC$ case, 
we want to minimize the number of edges outside the clusters.

We will use the online model, proposed by Charikar {\em et~al.}~\cite{CharikarCFM_incremental_2004}, 
and Mathieu {\etal}~\cite{MathieuSS_correlation_2010} for the online correlation clustering problem. 
Vertices (with their edges to previous vertices) arrive one 
at a time and must be clustered as soon as they arrive. 
Throughout the paper we will implicitly assume that any graph $G$ has its vertices ordered $v_1,\ldots,v_n$, 
according to the ordering in which they arrive on input.
The only two operations allowed are:
$\singleton(v_t)$, that creates a singleton cluster containing the single vertex $v_t$, and 
$\merge(C,C')$, which merges two existing clusters $C,C'$ into one, under the assumption that the resulting 
cluster induces a clique in $G$. This means that once two vertices are clustered together, 
they cannot be later separated.

For $\MaxCC$, we define the \emph{profit} of a clustering $\calC= \braced{C_1,\ldots,C_k}$ on a 
given graph $G=(V,E)$ to be the total number of edges in these cliques, that is 
$\sum_{i=1}^k \binom{|C_i|}{2}=\half \sum_{i=1}^k |C_i|(|C_i|-1)$.
Similarly, for $\MinCC$, we define the \emph{cost} of $\calC$ to be the total number of edges outside the cliques,
that is $|E|-\sum_{i=1}^k\binom{|C_i|}{2}$.
For a graph $G$, we denote the optimal profit or cost for $\MaxCC$ and $\MinCC$, respectively, by
$\profit_\smallOPT(G)$ and $\cost_\smallOPT(G)$.

It is common to measure the performance of an online strategy by its \emph{competitive ratio}. This ratio is defined as the worst case ratio between the profit/cost of the online strategy and the profit/cost of an offline optimal strategy, one that knows the complete input sequence in advance.  
More formally, for an online strategy $\calS$,  we define $\profit_{\calS}(G)$ to be the profit of $\calS$
when the input graph is~$G=(V,E)$ and, similarly, let $\cost_{\calS}(G)\defeq|E|-\profit_{\calS}(G)$ be the cost of $\calS$ on~$G$.

We say that an online strategy $\calS$ is \emph{$R$-competitive} for $\MaxCC$, if there is a constant $\beta$ such that
for any input graph $G$ we have
\begin{equation}
	 R\cdot\profit_{\calS}(G) + \beta \ge \profit_\smallOPT(G).
		\label{eqn: max competitive ratio}
\end{equation}
Similarly $\calS$ is \emph{$R$-competitive} for $\MinCC$, if there is a constant $\beta$ such that
for any input graph $G$ we have
\begin{equation}
	 \cost_{\calS}(G) \le R\cdot\cost_\smallOPT(G) + \beta.
		\label{eqn: min competitive ratio}
\end{equation}
The reason for defining the competitive ratio differently for maximization and minimization problems is to have all ratios 
being at least $1$.
The smallest $R$ for which a strategy $\calS$ is $R$-competitive is called the
(asymptotic) \emph{competitive ratio} of $\calS$.
The smallest $R$ for which $\calS$ is $R$-competitive with $\beta = 0$ is called the
\emph{absolute competitive ratio} of~$\calS$.
(If it so happens that these minimum values do not exist, in both cases the competitive ratio is 
actually defined by the corresponding infimum.)

Note that an online strategy does not know when the last vertex arrives and, as a consequence,
in order to be $R$-competitive, it needs to ensure that the 
corresponding bound, (\ref{eqn: max competitive ratio}) or (\ref{eqn: min competitive ratio}),
is valid after each step.


\section{Online Maximum Clique Clustering}
\label{sec: maxclustering}

In this section we study online $\MaxCC$, the clique clustering problem where
the objective is to maximize the number of edges within the cliques.
The main results here are upper and lower bounds for the competitive
ratio. For the upper bound, we give a strategy that uses a doubling technique to
achieve a competitive ratio of at most $15.646$. For the lower bound,
we show that no online strategy has a competitive ratio smaller than $6$.
Additional results include a competitive analysis of the greedy strategy and
a lower bound for doubling based strategies.


\subsection{The Greedy Strategy for Online $\MaxCC$}
\label{subsec: greedy_algorithm}



{\Greedy}, the \emph{greedy} strategy for $\MaxCC$, merges each input vertex with the largest current cluster that maintains the clique property.
This maximizes the increase in profit at this step.
If no such merging is possible the vertex remains in its singleton cluster. Greedy strategies are
commonly used as heuristics for a variety of online problems and can be shown to behave well for certain of them; e.g.~\cite{MathieuSS_correlation_2010}. We show that the solution of {\Greedy} can be far from optimal for $\MaxCC$.

For $n = 1,2,3$, {\Greedy} always finds an optimal clustering; see Figure~\ref{fig:greedysmall}, where all cases are shown. Therefore throughout the rest of this section
we will be assuming that $n\ge 4$.

\begin{figure}[htb]
\begin{center}
\includegraphics{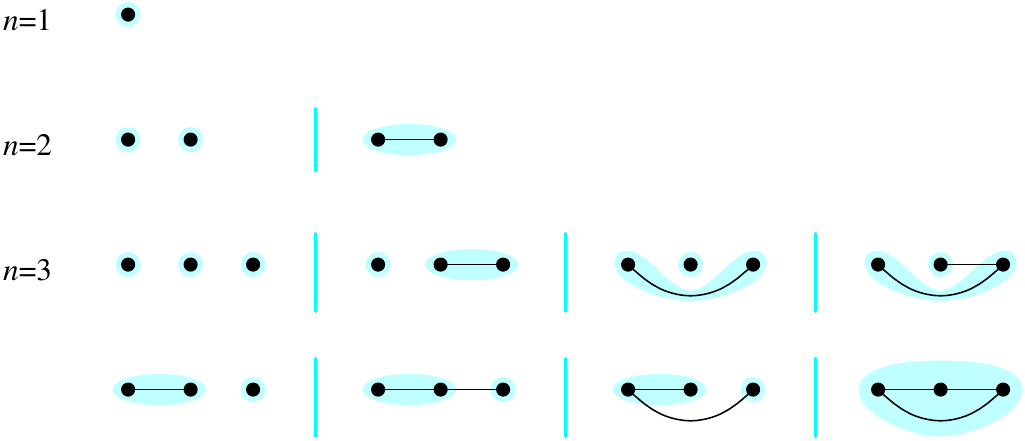}
\caption{\label{fig:greedysmall}Greedy finds optimal clusterings (blue) for $n = 1,2,3$. Vertices are released in order from left to right.}
\end{center}
\end{figure}


\begin{theorem}\label{thm:maxgreedylowerbound}
{\Greedy} has competitive ratio at least~$\floor{n/2}$ for $\MaxCC$.
\end{theorem}

\begin{proof}
We first give the proof for the absolute ratio, and then extend it to the asymptotic ratio.

Consider an adversary that provides input to the strategy to make it behave as badly as possible.
Our adversary creates an instance with $n$ vertices, numbered from $1$ to $n$. The odd vertices are connected to form a clique,
and similarly the even vertices are connected to form a clique. In addition each vertex of the form $2i$,
for $i=1,\ldots,\floor{(n-1)/2}$, is connected to vertex $2i-1$; see Figure~\ref{fig:greedyex}.

\begin{figure}[htb]
\begin{center}
\includegraphics{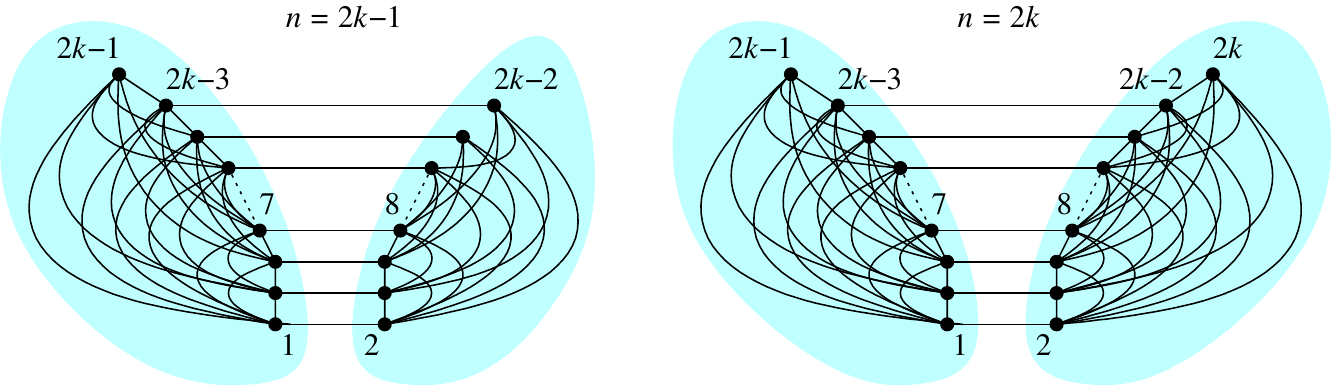}
\caption{\label{fig:greedyex}Illustrating the proof of Theorem~\protect\ref{thm:maxgreedylowerbound} for odd and even~$n$.}
\end{center}
\end{figure}

\Greedy\ clusters the vertices as odd/even pairs, leaving the vertex $2k-1$ as a singleton,
if $n=2k-1$ is odd and leaving both vertices $2k-1$ and $2k$ as singletons, if $n=2k$ is even.
This generates a clustering of profit $\profit_{\GDY}(G)=k-1$. An optimal strategy clusters
the odd vertices in one clique of size $k$ and the even vertices in another clique of size $k-1$ or $k$,
depending on whether $n$ is odd or even. The profit for the optimal solution
is $\profit_{\smallOPT}(G)= (k-1)^2$, if $n$ is odd and $\profit_{\smallOPT}(G)=k(k-1)$,
if $n$ is even. Hence, the ratio between the optimum and the greedy solution
is $k-1=(n-1)/2=\floor{n/2}$, if $n$ is odd, and $k=n/2=\floor{n/2}$, if $n$ is even; hence the worst case
absolute competitive ratio  of the greedy strategy is at least~$\floor{n/2}$.

To obtain the same lower bound on the asymptotic ratio, it suffices to notice that,
if we follow the above adversary strategy, then for any $R < \floor{n/2}$ and any constant $\beta> 0$, we can
find sufficiently large $n$ for which inequality~(\ref{eqn: max competitive ratio}) will be false.
\end{proof}

Next, we look at the upper bound for the greedy strategy.


\begin{theorem}\label{thm:maxgreedyupperbound}
{\Greedy}'s absolute competitive ratio for $\MaxCC$ is at most~$\floor{n/2}$.
\end{theorem}

\begin{proof}
As shown earlier, the theorem holds for $n=1,2,3$, so we can assume that $n\ge 4$.
Fix an optimal clustering on $G$ that we denote $\OPT(G)$.
Assume this clustering consists of $p$ non-singleton clusters
of sizes $c_1,\ldots,c_p$.
The profit of $\OPT(G)$ is $\profit_{\smallOPT}(G)= \half\sum_{i=1}^p c_i (c_i-1)$.
Let $k = \max_i c_i$ be the size of the maximum cluster of~$\OPT(G)$.

\medskip
\noindent
\mycase{1} $k \leq \floor{n/2}$. 
In this case, we can distribute the profit of each
cluster of \Greedy\ equally among the participating vertices; that is,
if a vertex belongs to a \Greedy\ cluster of size $c$, it will be assigned a profit of $\half(c-1)$.
We refer to this quantity as \emph{charged profit}.
We now note that at most one vertex in each cluster of $\OPT(G)$ can be a singleton cluster in
{\Greedy}'s clustering, since otherwise \Greedy\ would cluster any two such vertices together. 
This gives us that each
vertex in a non-singleton cluster of $\OPT(G)$, except possibly for one, has charged profit at least $\half$. So
the total profit charged to the vertices of an $\OPT(G)$ cluster of size $c_i$ is at least $\half(c_i-1)$.
Therefore the \emph{profit ratio} for this clique of $\OPT(G)$, namely the ratio between its optimal profit and
{\Greedy}'s charged profit, is at most
\begin{equation*}
\frac {\half c_i (c_i-1)} { \half (c_i - 1) } = c_i.
\end{equation*}
From this bound and the case assumption, all cliques of
$\OPT(G)$ have profit ratio at most $k\le \floor{n/2}$, so the competitive ratio is also at most $\floor{n/2}$.


\medskip
\noindent
\mycase{2} $k \ge \floor{n/2}+1$. In this case there is a unique
cluster $Q$ in $\OPT(G)$ of size $k$. 
The optimum profit is maximized if the graph has one other clique of size $n-k$, so
\begin{equation}
\profit_{\smallOPT}(G) \;\le\; \half k(k-1) + \half (n-k)(n-k-1)
					\;=\; \half (n^2 + 2k^2 - 2nk - n).
						\label{eqn: opt bound}
\end{equation}
We now consider two sub-cases.


\medskip
\noindent
\mycase{2.1} {\Greedy}'s profit is at least $k$. In this case, using~(\ref{eqn: opt bound})
and $k\ge \floor{n/2}+1 \ge \half (n+1)$, the competitive ratio is at most
\begin{equation*}
	\frac{\half (n^2 + 2k^2 - 2nk - n) } {k} \;\le\; \half (n-1) 
							\;\le\; \floor{n/2}.
\end{equation*}
where the first inequality follows from simple calculus. (The function
$(n^2 + 2k^2 - 2nk - n)- k(n-1)$ is non-positive for
$k = \half(n+1)$ and $k=n$ and its second derivative with respect to $k$
between these two values is positive.)


\medskip
\noindent
\mycase{2.2} {\Greedy}'s profit is at most $k-1$.
We show that in this case the profit of {\Greedy} is in fact \emph{equal to} $k-1$, and that
{\Greedy}'s clustering has a special form.

To prove this claim, consider those clusters of {\Greedy} that intersect $Q$. 
For $i\ge 1$ and $j\ge 0$, let $d_{ij}$ be the number of these clusters
that have $i$ vertices in $Q$ and $j$ outside $Q$. 
Note that at most one cluster of {\Greedy} can be wholly contained in $Q$, as
otherwise {\Greedy} would merge such clusters if it had more. 
Denote by $\alpha$ the size of
this cluster of {\Greedy} contained in $Q$ (if it exists; if not, let $\alpha = 0$).
Let also $\beta = d_{11}$ and 
\begin{equation*}
	\gamma \;=\; \sum_{\substack{i,j\geq1\\ i+j\geq 3}} \!\! id_{ij} \;=\; k - \alpha - \beta \;\geq\; 0,
\end{equation*}
where $k=\sum_{i\geq1,j\geq0} i d_{ij}$ counts the number of vertices in~$Q$.
The total profit of {\Greedy} is at least
\begin{align*}
 \half \!\! \sum_{i\ge 1, j \geq 0} \!\! (i+j)(i+j-1) d_{ij}
 		\; &=\;
		\half \alpha(\alpha-1) + \beta + \half \!\!\sum_{\substack{i,j\geq1 \\ i+j\geq 3}} \!\! (i+j)(i+j-1)d_{ij}
	\\ 
		&\ge\; 
		\half \alpha(\alpha-1) + \beta + \threehalves \!\!\sum_{\substack{i,j\geq1 \\ i+j\geq 3}} \!\! i d_{ij}
	\\ 
		& = \; \half \alpha(\alpha-1) + \beta + \threehalves \gamma
	\\
		&=\; k + \half \alpha(\alpha-3) + \half\gamma
		\;\ge\; k - 1 + \half\gamma.
\end{align*}
(The last inequality holds
because, for integer values of $\alpha$, the expression $\alpha (\alpha - 3)$ is minimized 
for $\alpha\in\braced{1,2}$.)
Combined with the case assumption that  {\Greedy}'s profit is at most $k-1$, 
we conclude that {\Greedy}'s profit is indeed equal $k-1$ and, in addition,
we have that $\gamma = 0$ and $\alpha\in\braced{1,2}$.

So, for $\alpha = 1$, {\Greedy}'s clustering consists of $k-1$ disjoint edges, each with 
exactly one endpoint in $Q$, plus a singleton vertex in $Q$. Thus $n\ge 2k-1$. 
As $k \ge \floor{n/2}+1$, this is possible only when $n = 2k-1$.
By~(\ref{eqn: opt bound}), the optimal profit in this case is at most $(k-1)^2$, so the ratio
is at most $k-1 = \floor{n/2}$.

For $\alpha = 2$, {\Greedy}'s clustering consists of $k-1$ edges, of which one is contained in $Q$ and 
the remaining ones have exactly one endpoint in $Q$. So $n \ge 2k-2$. If $n$ is odd, this and the bound
$k \ge \floor{n/2}+1$ would force $n = 2k-1$, in which case the argument from the paragraph above applies.
On the other hand, if $n$ is even, then these bounds will force $n = 2k-2$.
Then, by~(\ref{eqn: opt bound}), the optimal profit is $k^2-3k+3$, so the
competitive ratio is at most
$(k^2-3k+3)/(k-1) = k - 2 +1/(k-1) \leq k-1 = \floor{n/2}$, for~$k\ge 2$.
\end{proof}


\subsection{A Constant Competitive Strategy for $\MaxCC$}
\label{subsec: competitive_algorithm}



In this section, we give our competitive online strategy $\OCC$. Roughly, the strategy works
in phases. In each phase
we consider the ``batch'' of nodes that have not yet been clustered with other
nodes, compute an optimal clustering for this batch, and add these new clusters
to the strategy's clustering. The phases are defined so that the profit
for consecutive phases increases exponentially.

The overall idea can be thought of as an application of the 
``doubling'' strategy (see \cite{Chrobak_Mathieu_doubling_2006}, for example), but in our
case a subtle modification is required. Unlike other doubling approaches, in our
strategy the phases are not completely independent: the clustering computed in each
phase, in addition to the new nodes, needs to include the singleton nodes from
earlier phases as well. This is needed, because in our objective function singleton
clusters do not bring any profit. 

We remark that one could alternatively consider using profit value $k^2/2$ for a clique of size $k$,
which is a very close approximation to our function if $k$ is large. This would lead to a
simpler strategy and much simpler analysis. However, this function is a bad approximation when
the clustering involves many small cliques. This is, in fact, the most
challenging scenario in the analysis of our algorithm, and instances with this property are also
used in the lower bound proof.


\subsubsection{The Strategy~$\OCC$}

Formally, our method works as follows. Fix some constant parameter
$\gamma > 1$. The strategy works in phases, starting with phase $j=0$.
At any moment the clustering maintained by the strategy contains a set $U$ of \emph{singleton} clusters.  
During phase $j$,
each arriving vertex is added into $U$.  As soon as there is a clustering of $U$ of profit at least $\gamma^j$, 
the strategy clusters $U$ according to this clustering and
adds these new (non-singleton) clusters to its current clustering. (The vertices that form singleton clusters
remain in $U$.) Then phase $j+1$ starts.

Note that phase~$0$ ends as soon as one edge is revealed, since then it is possible for $\OCC$ to create 
a clustering with $\gamma^0=1$ edge. The last phase may not be complete; as a result all nodes
released in this phase will be clustered as singletons. Observe also that the strategy never
merges non-singleton cliques produced in different phases.


\subsubsection{Asymptotic Analysis of~$\OCC$}
\label{sec: analysis occ}

For the purpose of the analysis it is convenient to consider (without loss of generality) only
infinite ordered graphs $G$, whose vertices arrive one at a time in some order $v_1,v_2,\ldots$, and we
consider the ratios between the optimum profit and $\OCC$'s profit after each step.
Furthermore, to make sure that all phases are well-defined, we will assume that the optimum profit
for the whole graph $G$ is unbounded. Any finite instance can be converted into an infinite
instance with this property by appending to it an infinite sequence of disjoint edges, without
decreasing the worst-case profit ratio.

For a given instance (graph) $G$, define $\opt_j(G)$ to be the total profit of the adversary at the end of phase~$j$
in the $\OCC$'s computation on $G$.
Similarly, $\alg_j(G)$ denotes the total profit of Strategy~{\OCC} at the end of phase~$j$ (including the incremental
clustering produced in phase~$j$). During phase $0$ the graph is empty, and
at the end of phase $0$ it consists of only one edge, so $\alg_0(G) = \opt_0(G) = 1$. 
For any phase $j > 0$, the profit of $\OCC$ is equal to~$\alg_{j-1}(G)$ throughout the phase,
except right after the very last step, when new non-singleton clusters are created.
At the same time, the optimum profit can only increase. Thus the maximum ratio in
phase $j$ is at most $\opt_j(G)/\alg_{j-1}(G)$. We can then conclude that,
to estimate the competitive ratio of our strategy $\OCC$, it is sufficient to establish an asymptotic
upper bound on numbers $\R_j$, for $j=1,2,\ldots$, defined by
\begin{equation}
	\R_j \;=\; \max_G \frac{\opt_j(G)}{\alg_{j-1}(G)},
	\label{eqn: definition of Rj}
\end{equation}
where the maximum is taken over all infinite ordered graphs $G$. (While not immediately obvious,
the maximum is well-defined. There are infinitely many prefixes of $G$ on which $\OCC$ 
will execute $j$ phases, due to the presence of singleton clusters. However, since these singletons induce an
independent set after $j$ phases, only finitely many graphs $G$ need to be considered in this maximum.)

Our objective now is to derive a recurrence relation for the sequence $\R_1,\R_2,\ldots$.
The value of $\R_1$ is some constant whose exact
value is not important here since we are interested in the asymptotic ratio. (We will, however, estimate
$\R_1$ later, when we bound the absolute competitive ratio in Section~\ref{subsubsec: absolute_ratio}).

\medskip

So now, assume that $j\ge 2$ and that $\R_1,\R_2,\ldots,\R_{j-1}$ are given. We want to bound
$\R_j$ in terms of $\R_1,\R_2,\ldots,\R_{j-1}$. To this end, let $G$ be the graph for which
$\R_j$ is realized, that is $\R_j = \opt_j(G)/\alg_{j-1}(G)$.
With $G$ fixed, to avoid clutter, we will omit it in our notation, writing
$\opt_j = \opt_j(G)$, $\alg_j = \alg_{j}(G)$, etc. In particular, $\R_j = \opt_j/\alg_{j-1}$.

We now claim that, without loss of generality, we can assume that in the
computation on $G$, the incremental
clusterings of Strategy~{\OCC} in each phase $1,2,\ldots,j-1$ do not contain any
singleton clusters. (The clustering in phase $j$, however, is allowed to contain singletons.)
We will refer to this property as the \emph{No-Singletons Assumption}.

To prove this claim, we modify the ordering of $G$ as follows: 
if there is a phase $i < j$ such that the incremental clustering of
$U$ in phase~$i$ clusters some vertex $v$ from $U$ as a singleton, then
delay the release of $v$ to the beginning of phase~$i+1$.
Postponing a release of a vertex that was clustered as a singleton in 
some phase $i < j$ to the beginning of phase $i+1$ does not affect the
computation and profit of $\OCC$, because vertices from singleton clusters remain in $U$,
and thus are available for clustering in phase $i+1$.
In particular, the value of $\alg_{j-1}$ will not change.
This modification also does not change the value of $\opt_j$, because the graph induced by the
first $j$ phases is the same, only the ordering of the vertices has been changed.
We can thus repeat this process until the No-Singletons Assumption is eventually satisfied.
This proves the claim.

With the No-Singletons Assumption, the set $U$ is empty at the beginning of each phase $0,1,\ldots,j$.
We can thus divide the vertices released in phases $0,1,\ldots,j$ 
into  disjoint \emph{batches}, where batch $B_i$ contains the vertices released in phase~$i$, for $i=0,1,\ldots,j$. 
(At the end of phase $i$, right before the clustering is updated, we will have $B_i = U$.)
For each such $i$, denote by $\deltaalg_i$ the maximum profit of a clustering of $B_i$.
Then the total profit after $i$ phases is $\alg_i = \deltaalg_0 + \cdots + \deltaalg_i$, and, 
by the definition of $\OCC$, we have
$\deltaalg_i  \geq  \gamma^i$ and $\alg_i \geq  (\gamma^{i+1}-1)/(\gamma-1)$.

For $i=0,1,\ldots,j$,
let $\barB_i = B_0\cup\ldots\cup B_i$ be the set of all vertices released in phases~$0,\ldots,i$.
Consider the optimal clustering of $\barB_j$.  In this clustering, every cluster has some number $a$ of
nodes in $\barB_{j-1}$ and some number $b$ of nodes in $B_j$. 
For any $a,b\ge 0$, let $k_{a,b}$ be the number of clusters of this form in the optimal clustering of $\barB_j$.
Then we have the following bounds, where the sums range over all integers $a,b\geq 0$.
\begin{align}
	\opt_j &=  \sum \profvalue{a+b}k_{a,b}
	\label{eqn: formula for optj}
	\\
	\opt_{j-1} &\ge  \sum \profvalue{a} k_{a,b}
	\label{eqn: bound for optj-1}
	\\
	\deltaalg_j &\ge  \sum \profvalue{b} k_{a,b}
	\label{eqn: bound for Deltajalg}
	\\
	\alg_{j-1} &\ge  \half \sum ak_{a,b}
	\label{eqn: bound for algj-1}
\end{align}
Equality (\ref{eqn: formula for optj}) is the definition of $\opt_j$.
Inequality (\ref{eqn: bound for optj-1}) holds 
because the right hand side represents the profit of the optimal clustering of $\barB_j$ restricted to $\barB_{j-1}$,
so it cannot exceed the optimal profit $\opt_{j-1}$ for $\barB_{j-1}$.
Similarly, inequality (\ref{eqn: bound for Deltajalg}) holds
because the right hand side is the profit of the optimal clustering of $\barB_j$ restricted to $B_j$,
while $\deltaalg_j$ is the optimal profit of $B_j$.
The last bound (\ref{eqn: bound for algj-1}) follows from the fact that (as a consequence of the No-Singletons Assumption)
our strategy does not have any singleton clusters in $\barB_{j-1}$. 
This means that in {\OCC}'s clustering of $\barB_{j-1}$ 
(which has $\sum ak_{a,b}$ vertices) each
vertex has an edge included in some cluster, so the number of these edges must be at
least $\half \sum a k_{a,b}$.

We can also bound $\deltaalg_j$, the strategy's profit increase, from above. We have
$\deltaalg_0 = 1$ and for each phase $j\geq1$,
\begin{equation}
	\deltaalg_j \;\le\;  \gamma^j + \half (\sqrt{8\gamma^j + 1}+1 )  \;<\; \gamma^j + \sqrt{2}\gamma^{j/2} +2-\sqrt{2}.
	\label{eqn: upper bound for Deltajalg}
\end{equation}
To show (\ref{eqn: upper bound for Deltajalg}), suppose that phase $j$ ends at step
$t$ (that is, right after $v_t$ is revealed). Consider the optimal partitioning $\calP$ of $B_j$, 
and let the cluster $c$ containing $v_t$ in $\calP$ have size $p+1$.
If we remove $v_t$ from this partitioning, we obtain a partitioning $\calP'$ of the batch
after step $t-1$, whose profit must be strictly smaller than $\gamma^j$.
So the profit of $\calP$ is smaller than $\gamma^j+p$.
In partitioning $\calP'$, the cluster $c-\{v_t\}$ has size $p$. We thus
obtain that $\binom{p}{2} < \gamma^j$, because, in the worst case, $\calP$ consists only of the cluster $c$.
This gives us $p< \half(\sqrt{8\gamma^j + 1}+1)$.
The second inequality in~(\ref{eqn: upper bound for Deltajalg}) follows by routine calculation.

From~(\ref{eqn: upper bound for Deltajalg}), by adding up all profits from phases
$0,\ldots,j$, we obtain an upper bound on the total profit of the strategy,
\begin{equation}
	\alg_{j} < \frac{\gamma^{j+1}-1}{\gamma-1} 
				+ \sqrt{2}\cdot\frac{\gamma^{(j+1)/2}-\gamma^{1/2}}{\gamma^{1/2}-1} + (2-\sqrt{2})j + 1.
	\label{eqn: upper bound for algj}
\end{equation}
%


\begin{lemma}\label{lem:profvalue}
For any pair of non-negative integers $a$ and $b$, the inequality 
\begin{equation*}
\profvalue{a+b} \leq (x+1)\profvalue{a} + \frac{x+1}{x}\profvalue{b} + a
\end{equation*}
holds for any $0<x\leq1$.
\end{lemma}

\begin{proof}
Define the function
\begin{align*}
F(a,b,x) & = 2x(x+1)\profvalue{a} + 2(x+1)\profvalue{b} + 2ax - 2x\profvalue{a+b}
\\
	& =  a^2x^2 - ax^2 + 2ax + b^2 - b - 2abx
\\
	& = (b-ax)^2 + ax(2-x) - b,
\end{align*}
i.e., $2x$ times the difference between the right hand side and the left hand side of the inequality above.
It is sufficient to show that $F(a,b,x)$ is non-negative for integers $a,b\geq0$ and $0<x\leq1$.

Consider first the cases when $a\in \{0,1\}$ or $b\in\{0,1\}$.
$F(0,b,x)=b(b-1)\geq0$, for any non-negative integer $b$ and any~$x$.
$F(a,0,x)=ax(ax-x+2)\geq ax(ax+1)>0$, for any positive integer $a$ and $0<x\leq1$.
$F(a,1,x)=x^2a(a-1)\geq0$, for any positive integer $a$ and any~$x$.
$F(1,2,x)=2-2x\geq0$, for $0<x\leq1$, and $F(1,b,x)=b^2-b+2x-2bx\geq b^2-3b\geq0$, for any integer $b\geq3$ and $0<x\leq1$.

Thus, it only remains to show that $F(a,b,x)$ is non-negative when both $a\geq2$ and~$b\geq2$.
The function $F(a,b,x)$ is quadratic in $x$ and hence has one local minimum at $x_0=\frac{b-1}{a-1}$, 
as can be easily verified by differentiating $F$ in~$x$.
Therefore, in the case when $a\leq b$, $F(a,b,x)\geq F(a,b,1)=(b-a)^2 - (b-a)\geq 0$, for $0<x\leq1$.
In the case when $a>b$, we have that $F(a,b,x)\geq F(a,b,\frac{b-1}{a-1})=\frac{(a-b)(b-1)}{a-1}>0$, which completes the proof.
\end{proof}

Suppose that $j\ge 2$ and fix some parameter $x$, $0 < x < 1$, whose value we will determine later.
Using Lemma~\ref{lem:profvalue}, the bounds (\ref{eqn: formula for optj})--(\ref{eqn: bound for algj-1}),
and the definition of $\R_{j-1}$, we obtain
\begin{align}
\R_j\alg_{j-1} = \opt_j 
& =
\sum \profvalue{a+b} k_{a,b} \nonumber
\\
& \leq
(x+1)\sum \profvalue{a}k_{a,b} + \frac{x+1}{x}\sum\profvalue{b}k_{a,b} + \sum a k_{a,b} \nonumber
\\
& \leq
(x+1) \opt_{j-1} + \frac{x+1}{x} \deltaalg_j  + 2 \alg_{j-1} \label{eqn: opt recurrence}
\\
& \le  
(x+1)\R_{j-1} \alg_{j-2} + \frac{x+1}{x} \deltaalg_j  + 2 \alg_{j-1}. \nonumber
\end{align}
Thus $\R_j$ satisfies the recurrence
\begin{equation}
	\R_j \leq	\frac{ x+1}{x \alg_{j-1}}
	 			\left[\rule[0pt]{0pt}{3ex}\,x\alg_{j-2} \R_{j-1} +\deltaalg_j\,\right]	 + 2.
				\label{eqn: recurrence for Rj}
\end{equation}
From inequalities~(\ref{eqn: upper bound for Deltajalg})~and~(\ref{eqn: upper bound for algj}),
we have 
\begin{equation*}
	\deltaalg_i = \gamma^i(1+o(1))
	\quad \textrm{and} \quad
	\alg_i = \frac{\gamma^{i+1}(1+o(1))}{\gamma-1}.
\end{equation*}
for all $i = 0,1,\ldots,j$. Above, we use the notation $o(1)$ to denote any function that tends to $0$ as the
phase index $i$ goes to infinity (with $x$ and $\gamma$ assumed to be some fixed constants, still to be determined).
Substituting into  recurrence (\ref{eqn: recurrence for Rj}), we get
\begin{equation}
	\R_j \leq	\Big( \frac{x+1}{\gamma} +o(1) \Big)\cdot \R_{j-1}
				+ \frac{(x+1)(\gamma-1)}{x} + 2 + o(1).
				\label{recurrence for Rj, subst}
\end{equation}
Now define
\begin{equation}
 	\R = \frac{\gamma ( \gamma x + x + \gamma - 1 ) }{ x (\gamma -x - 1)}.
	\label{eqn: definition of R}
 \end{equation}
%


\begin{lemma}\label{lem: bound on Rj's}
Assume that $x+1 < \gamma$, then $\R_j = \R+o(1)$.
\end{lemma}

\begin{proof}
The proof is by routine calculus, so we only provide a sketch. For all $j\ge 1$ let 
$\rho_j = \R_j - \R$.
Then, substituting this into (\ref{recurrence for Rj, subst}) and simplifying, we obtain
that the $\rho_j$'s satisfy the recurrence
\begin{equation}
 	\rho_j \leq \Big( \frac{x+1}{\gamma} +o(1) \Big)\cdot  \rho_{j-1} + o(1).
 \end{equation}
Since $x+1 < \gamma$, this implies that $\rho_j = o(1)$, and the lemma follows.
\end{proof}

Lemma~\ref{lem: bound on Rj's} gives us (essentially) a bound of $\R$ on the asymptotic
competitive ratio of Strategy~{\OCC}, for fixed
values of parameters $\gamma$ (of the strategy) and $x$ (of the analysis).
We can now choose $\gamma$ and $x$ to make $\R$ as small as possible.
$\R$ is minimized for parameters $x= \half (5-\sqrt{13}) \approx 0.697$ and 
$\gamma= \half (3+\sqrt{13}) \approx 3.303$, yielding 
\begin{equation*}
 \R = \onesixth (47+13\sqrt{13}) \approx 15.646.
\end{equation*}
Using Lemma~\ref{lem: bound on Rj's}, for each graph $G$ and phase $j$, we have that
$\opt_j(G) \le (\R+o(1))\alg_{j-1}(G)$. Since, in fact, $\R < 15.646$, this implies that 
$\opt_j(G) \le 15.646 \cdot \alg_{j-1}(G)$, as long as $j$ is large enough.
Thus $\opt_j(G) \le 15.646 \cdot \alg_{j-1}(G) + O(1)$ for all phases $j$.
As we discussed earlier, bounding $\opt_j(G)$ in terms of $\alg_{j-1}(G)$ like this
is sufficient to establish a bound on the (asymptotic) competitive ratio of Strategy~$\OCC$.
Summarizing, we obtain the following theorem.
\begin{theorem}\label{thm: occ upper bound}
The asymptotic competitive ratio of Strategy~$\OCC$ is at most~$15.646$.
\end{theorem}


\subsubsection{Absolute Competitive Ratio}
\label{subsubsec: absolute_ratio}


In fact, for $\gamma= \half (3+\sqrt{13})$, Strategy~{\OCC} has a low absolute competitive ratio as well.
We show that this ratio is at most $22.641$. The argument uses the same value of
parameter $x= \half (5-\sqrt{13})$, but requires a more refined analysis.

When phase~0 ends, the competitive ratio is $1$. For $j\ge 1$, let $\opt'_j$ 
be the optimal profit right before phase $j$ ends. (Earlier we used
$\opt_j$ to estimate this value, but $\opt_j$ also includes the profit for
the last step of phase~$j$.)
It remains to show that for phases $j\ge 1$ we have $\R'_j \le 22.641$,
where $\R'_j = \opt'_j/\alg_{j-1}$.

By exhaustively analyzing the behavior of Strategy~$\OCC$ in phase~1, taking into account that
$\gamma\approx 3.303 > 3$, we can establish that $\R'_1=10$.
We will then bound the remaining ratios using a refined version of recurrence~(\ref{eqn: recurrence for Rj}).

We start by estimating $\R'_1$.
Let $t$ be the last step of phase~1. Since $\gamma\approx 3.303$, after step $t-1$ the profit of the vertices released in phase~1 is at most $3$. We can assume that phase~0 has only two vertices $v_1,v_2$ connected by an edge. Let $H$ be the graph induced by $v_1,\ldots,v_{t-1}$ and $H'$ be its subgraph induced by $v_3,\ldots,v_{t-1}$. We thus want to bound the optimal profit of $H$, under the assumption that the optimal profit of $H'$ is at most~$3$.

Denote by $\KK_i$ the clique with $i$ vertices. The optimal clustering of $H'$ cannot include a $\KK_4$, and either
\begin{enumerate} 
\item\label{case:noK3}
$H'$ has no $\KK_3$, and it has at most three $\KK_2$ cliques, or
\item\label{case:hasK3}
$H'$ has a $\KK_3$, with each edge of $H'$ having at least one endpoint in this $\KK_3$.
\end{enumerate}
In Case~\ref{case:noK3}, $H$ cannot contain a $\KK_5$. If a clustering of $H$ includes a $\KK_4$ then this $\KK_4$ contains $v_1$, $v_2$,
and two vertices from phase~$1$. So in addition to this $\KK_4$ it can at best include two $\KK_2$'s, for a total profit of at most $8$. 
In Case~\ref{case:hasK3}, if a clustering of $H$ includes a $\KK_5$, then it cannot include any cluster except this $\KK_5$, 
so its profit is $10$. If a clustering of $H$ includes a $\KK_4$, then this $\KK_4$ must contain at least one of $v_1$ and $v_2$, 
and it may include at most one other clique of type $\KK_2$. This will give a total profit of at most $7$. 
Summarizing, in each case the profit of $H$ is at most $10$ giving us $\R'_1\leq10$, as claimed.

For phases $j\geq2$, we can tabulate upper bounds for $\R'_j$ by explicitly
computing the ratios $\R'_j=\opt'_j/\alg_{j-1}$ using the following modification of
recurrence~(\ref{eqn: recurrence for Rj}), 
\begin{equation}
	\R'_j \leq	\frac{ x+1}{x \alg_{j-1}}
	 			\left[\rule[0pt]{0pt}{3ex}\,
	 			x\alg_{j-2} \R'_{j-1} +\deltaalg_j\,\right]	 + 2,
				\label{eqn: recurrence for R'j}
\end{equation}
where we use the more exact bounds
\begin{equation*}
	\lceil\gamma^j\rceil\leq  \deltaalg_j  \leq \floor{ \gamma^j + \half( \sqrt{8\gamma^j+1} + 1 )},
\end{equation*}
obtained by rounding the bounds $\deltaalg_j \ge \gamma^j$ and~(\ref{eqn: upper bound for Deltajalg}), which we can do
because $\deltaalg_j$ is integral.
From the definition of $\alg_j \defeq 1 + \sum_{i=1}^j \deltaalg_i$ we compute the first few estimates as shown in Table~\ref{tab:values}.

\begin{table}
\caption{Some initial bounds for $\alg_j$ and the absolute competitive ratio.\label{tab:values}}
\begin{center}
\begin{tabular}{|l|r|r|r|r|r|r|r|r|r|} \hline
Phase ($j$) & 0 & 1 & 2 & 3 & 4 & 5 & 6 & 7 & 8 
\\ \hline
min $\alg_j$
    & 1
    & 5
    & 16
    & 53
    & 172
    & 566
    & 1~864
    & 6~152
    & 20~311
\\ \hline
max $\alg_j$
    & 1
    & 7
    & 23
    & 68
    & 202
    & 623
    & 1~972
    & 6~352
    & 20~679
\\ \hline
Bound ($\R'_j$)
    & 1.000
	& 10.000
	& 13.185
	& 18.636
	& 21.881
	& {\bf 22.641}
	& 21.516
	& 19.925
	& 18.509
\\ \hline
\end{tabular}
\end{center}
\end{table}

To bound the sequence $\{\R'_j\}_{j\geq9}$ we rewrite recurrence~(\ref{eqn: recurrence for R'j}) as
\begin{equation*}
\R'_j \leq
		\frac{(x+1)\alg_{j-2}}{\alg_{j-1}} 
			\cdot	\R'_{j-1} + \frac{(x+1)\deltaalg_{j}}{x\alg_{j-1}} + 2
		=
			\alpha_j \R'_{j-1} + \beta_j,
\end{equation*}
and bound $\alpha_j$ and $\beta_j$ using (\ref{eqn: upper bound for Deltajalg}) and~(\ref{eqn: upper bound for algj}). With routine calculations, we can establish the bounds $\alpha_j< \threefifths$ and $\beta_j<8$, for $j\geq8$.

Thus, $\R'_j\leq\hat{\R}_j$, where $\hat{\R}_j$ is 
\begin{equation*}
\hat{\R}_j = \threefifths \hat{\R}_{j-1} + 8 \leq 20 - a\left(\threefifths\right)^j,
\end{equation*}
for $j\geq8$ and some positive constant $a$. The sequence $\{\hat\R_j\}_{j\geq9}$, is thus bounded above by a monotonically growing function of $j$ having limit $20$ and hence $\hat\R_j\leq20$ for every $j\geq9$.

Combining this with the bounds estimated in Table~\ref{tab:values}, we see that the largest bound on $\R'_j$ is $22.641$ given for $j=5$. 
We can thus conclude that the absolute competitive ratio of \OCC\ is at most~$22.641$.

We can improve on the absolute competitive ratio by choosing different values for $\gamma$ and $x$ that allow the asymptotic competitive ratio to increase slightly. 
The optimal values can be found empirically (using mathematical software) to be $\gamma=4.02323428$ and $x=0.823889$,
giving asymptotic competitive ratio 15.902 and absolute competitive ratio~20.017.


\subsection{A Lower Bound for Strategy~$\OCC$}
\label{subsec: algOCC_lower_bound}



In this section we will show that, for any choice of $\gamma$, the worst-case ratio of Strategy~$\OCC$ is at least~$10.927$. 

Denote by $B_j$ the $j$-th batch, that is the vertices released in phase $j$. 
We will use notation $\alg_j$ for the profit of $\OCC$ and $\opt_j$ for the optimal profit on the sub-instance consisting of the first $j$ batches. To avoid clutter we will omit lower order terms in our calculations. In particular, we focus on $j$ being large enough, treating $\gamma^j$ as integer, and all estimates for $\alg_j$ and $\opt_j$ given below are meant to hold within a factor of $1 \pm o(1)$. (The asymptotic notation is with respect to the phase index $j$ tending to~$\infty$.)

We start with a simpler construction that shows a lower bound of $9$; then we will explain how to improve it to $10.927$. In the instance we construct, all batches will be disjoint, with the $j$th batch $B_j$ having $2\gamma^j$ vertices connected by $\gamma^j$ disjoint edges (that is, a perfect matching).  We will refer to these edges as \emph{batch} edges. The edges between any two batches $B_i$ and $B_j$, for $i<j$, form a complete bipartite graph. These edges will be called \emph{cross}~edges; see Figure~\ref{fig: example ratio}.


\begin{figure}
\begin{center}
\includegraphics{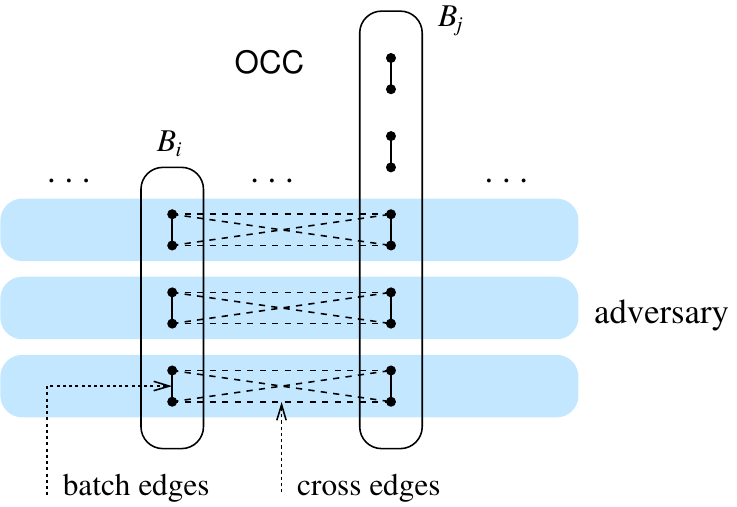}
\end{center}
\caption{The lower bound example for Strategy~$\OCC$. The figure shows two batches $B_i$ and $B_j$, for $i < j$. Batch edges, drawn with solid lines, are collected by Strategy~$\OCC$. Dashed lines show cross edges that are in the adversary's clustering. Shaded regions illustrate the cliques in the adversary's clustering.}
\label{fig: example ratio}
\end{figure}


At the end of each phase $j$, the strategy will collect all $\gamma^j$ edges inside $B_j$. Therefore, by summing up the geometric sequence, right before the end of phase $j$ (before the strategy adds the new edges from $B_j$ to its clustering), the strategy's profit is
\begin{equation*}
\alg_{j-1} = \sum_{i=0}^{j-1}\gamma^i \le \frac{\gamma^j}{\gamma-1}.
\end{equation*}
After the first $j$ phases, the adversary's clustering consists of cliques $C_p$, $p=0,1,\ldots,\gamma^j-1$, where $C_p$ contains the $p$-th edge (that is, its both endpoints) from each batch $B_i$ for $i= p,p+1,\ldots,j$; see Figure~\ref{fig: example ratio}. We claim that the adversary gain after $j$ phases satisfies
\begin{equation}
	\opt_j \ge \opt_{j-1} + \gamma^j + 4\sum_{i=0}^{j-1} \gamma^i
		= \opt_{j-1} + \frac{(\gamma+3)\gamma^j}{\gamma-1}.
		\label{qen: bad example opt_j}
\end{equation}
(Recall that all equalities and inequalities in this section are assumed to hold only within a factor of $1 \pm o(1)$.)
We now justify this bound. The second term $\gamma^j$ is simply the number of batch edges in $B_j$. To see where each term $4\gamma^i$ comes from, consider the $p$-th batch edge from $B_i$, for $i < j$. When we add $B_j$ after phase $j$, the adversary can add the $4$ cross edges connecting this edge's endpoints to the endpoints of the $p$th batch edge in $B_j$ to $C_p$. Overall, this will add  $4\gamma^i$ cross edges between $B_i$ and $B_j$ to the existing adversary's cliques.

From recurrence (\ref{qen: bad example opt_j}), by simple summation, we get
\begin{equation*}
	\opt_j \ge \frac{(\gamma+3)\gamma^{j+1}}{(\gamma-1)^2}.
\end{equation*}
Dividing it by $\OCC$'s profit of at most $\gamma^j/(\gamma-1)$, we obtain that the ratio is at least $\frac{\gamma(\gamma+3)}{\gamma-1}$, which, by routine calculus, is at least~$9$.

\smallskip

We now outline an argument showing how to improve this lower bound to $10.927$. The new construction is almost identical to the previous one, except that we change the very last batch $B_j$. As before, each batch $B_i$, for $i < j$, has $\gamma^i$ disjoint edges. Batch $B_j$ will also have $\gamma^j$ edges, but they will be grouped into 
$q = \onethird \gamma^j$ disjoint triangles. (So $B_j$ has $\gamma^j$ vertices.)
For $p = 0,1,\ldots,q-1$, we add the $p$-th triangle to clique $C_p$. (If $q > \gamma^{j-1}$, the last $q-\gamma^{j-1}$ triangles will form new cliques.)

This modification will preserve the number of edges in $B_j$ and thus it will not affect the strategy's profit. But now, for each $i=0,1,\ldots,j-1$ and each $p = 0,1,\ldots,\min(q,\gamma^{i})-1$,  we can connect the two vertices in $B_i\cap C_p$ to three vertices in $B_j$, instead of two. This creates two new cross edges that will be called  \emph{extra} edges. It should be intuitively clear that the number of these extra edges is $\Omega(\gamma^j)$, which means that this new construction gives a ratio strictly larger than~$9$.

Specifically, to estimate the ratio, we will distinguish three cases, depending on the value of $\gamma$. Suppose first that $\gamma \ge 3$. Then $q \ge \gamma^{j-1}$, so the number of extra edges is $2\sum_{i=0}^{j-1}\gamma^i = 2\gamma^j/(\gamma-1)$, because each vertex in $B_0\cup B_1\cup\cdots\cup B_{j-1}$ is now connected to three vertices in $B_j$, not two. Thus the new optimal profit is
\begin{equation*}
	\opt'_j = \opt_j + \frac{2\gamma^j}{\gamma-1} = \frac{(\gamma^2+5\gamma-2)\gamma^j}{(\gamma-1)^2}.
\end{equation*}
Dividing by $\OCC$'s profit, the ratio is at least $\frac{\gamma^2+5\gamma-2}{\gamma-1}$, which is at least $11$ for $\gamma\geq3$.

The second case is when $\sqrt{3}\le \gamma \le 3$. Then $\gamma^{j-2}\le q \le \gamma^{j-1}$. 
In this case all vertices in $B_0\cup B_1\cup \cdots \cup B_{j-2}$ and $\twothirds \gamma^j$ vertices in $B_{j-1}$ get an extra edge, so the number of extra edges is $2\gamma^{j-1}/(\gamma-1) + \twothirds \gamma^j$. Therefore the new adversary profit is
\begin{equation*}
	\opt'_j = \opt_j + 2\,\frac{\gamma^{j-1}}{\gamma-1} + \twothirds \gamma^j = \frac{(5\gamma^3 + 5\gamma^2 + 8\gamma - 6)\gamma^{j-1}}{3(\gamma-1)^2}.
\end{equation*}
We thus have that the ratio is at least $\frac{5\gamma^3 + 5\gamma^2 + 8\gamma - 6}{3\gamma(\gamma-1)}$. Minimizing this quantity, we obtain that the ratio is at least~$10.927$.

The last case is when $1 < \gamma \le \sqrt{3}$. In this case, even using the earlier strategy (without any extra edges), 
we have that the ratio $\opt_j/\alg_{j-1} = \frac{ \gamma(\gamma+3)}{ \gamma-1}$ is at least $3+6\sqrt{3}\approx 11.2$ (it is minimized for $\gamma = \sqrt{3}$).



\subsection{A Lower Bound of~6 for $\MaxCC$}
\label{subsec: lower_bound_6}




We now prove that any deterministic online strategy $\calS$ for the clique clustering problem has competitive ratio at least $6$. 
We present the proof for the absolute competitive ratio and explain later how to extend it to the asymptotic ratio. The lower bound is established by
showing, for any constant $R < 6$, an adversary strategy for constructing an input graph $G$
on which $\profit_{\smallOPT}(G) \ge R\cdot \profit_{\calS}(G)$, that is
the optimal profit is at least $R$ times the profit of~$\calS$.


\paragraph{Skeleton trees.}
Fix some non-negative integer $D$. (Later we will make the value of $D$ depend on $R$.)
It is convenient to describe the graph constructed by the adversary in terms of
its underlying \emph{skeleton tree} $\calT$, which is a rooted binary tree. The
root of $\calT$ will be denoted by $r$.
For a node $v\in \calT$, define the \emph{depth} of $v$ to be the
number of edges on the simple path from $v$ to $r$.
The adversary will only use skeleton trees of the following special form:
each non-leaf node at depths $0,1,\ldots,D-1$ has two children, 
and each non-leaf node at levels at least $D$
has one child. Such a tree $\calT$ can be thought of as consisting of its \emph{core subtree},
which is the subtree of $\calT$ induced by the nodes of depth up to $D$, 
with paths attached to its leaves at level $D$. 
The nodes of $\calT$ at depth $D$ are the leaves of the core subtree. If $v$ is a leaf
of the core subtree of $\calT$ then the path extending from $v$ down to a leaf of $\calT$
is called a \emph{tentacle} -- see Figure~\ref{fig: skeleton tree}. (Thus $v$ belongs
both to the core subtree and to the tentacle attached to $v$.)
The length of a tentacle is the number of its edges.
The nodes in the tentacles are all considered to be left children of their parents.


\begin{figure}[t]
\begin{center}
\includegraphics[width=4.25in]{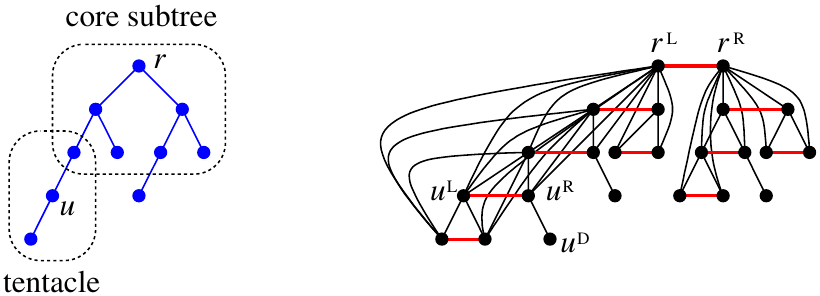}
\end{center}
\caption{On the left, an example of a skeleton tree $\calT$. The core subtree of
$\calT$ has depth $2$ and two tentacles, one of length $2$ and one of length $1$. 
On the right, the corresponding graph $\calG$.}
\label{fig: skeleton tree}
\end{figure}


\paragraph{Skeleton-tree graphs.}
The graph represented by a skeleton tree $\calT$ will be denoted by $\calG$.
We differentiate  between the \emph{nodes} of  $\calT$ and the \emph{vertices} of $\calG$. 
The relation between $\calT$ and $\calG$ is illustrated in Figure~\ref{fig: skeleton tree}.
The graph $\calG$ is obtained from the tree $\calT$ as follows:
\begin{itemize}
	\item For each node $u \in\calT$ we create two vertices
	$u^\LL$ and $u^\RR$ in $\calG$, with an edge between them. This edge $(u^\LL,u^\RR)$ is called
	the \emph{cross edge} corresponding to $u$.
	\item Suppose that $u,v\in\calT$. If $u$ is in the left subtree of $v$ then
	$(u^\LL,v^\LL)$ and  $(u^\RR,v^\LL)$ are edges of $\calG$.
	If $u$ is in the right subtree of $v$ then
	$(u^\LL,v^\RR)$ and  $(u^\RR,v^\RR)$ are edges of $\calG$. These edges are
	called \emph{upward edges}.
	\item If $u\in \calT$ is a node in a tentacle of $\calT$ and is not a leaf of $\calT$,
	then $\calG$ has a vertex $u^\DD$ with edge $(u^\DD,u^\RR)$. This edge
	is called a \emph{whisker}.
\end{itemize}


\paragraph{The adversary strategy.}
The adversary constructs $\calT$ and $\calG$ gradually, in response to strategy $\calS$'s choices.
Initially, $\calT$ is a single node $r$, and thus
$\calG$ is a single edge $(r^\LL,r^\RR)$. 
At this time, $\profit_\calS(\calT) = 0$ and $\profit_{\smallOPT}(\calT) = 1$, so
$\calS$ is forced to collect this edge (that is, it creates a $2$-clique
$\braced{ r^\LL,r^\RR }$), since otherwise the adversary can immediately stop with unbounded absolute competitive ratio. 

In general, the invariant of the construction is that, at each step, the only non-singleton cliques
that $\calS$ can add to its clustering are cross edges that correspond to the current leaves of $\calT$.
Suppose that, at some step, $\calS$ collects a cross edge $(u^\LL,u^\RR)$, corresponding to 
node $u$ of $\calT$. ($\calS$ may collect more cross edges in one step; if so, the adversary
applies its strategy to each such edge independently.)
If $u$ is at depth less than $D$, the adversary extends $\calT$ by adding two children of $u$.
If $u$ is at depth at least $D$, the adversary only adds the left child of $u$, thus extending
the tentacle ending at $u$.
In terms of $\calG$, the first move appends two triangles to $u^\LL$ and $u^\RR$, with
all corresponding upward edges. The second move appends a triangle to $u^\LL$ and a whisker to $u^\RR$
(see Figure~\ref{fig: adversary moves}).
In the case when $\calS$ decides not to collect any cross edges at some step, the adversary stops the
process.


\begin{figure}
\begin{center}
\includegraphics[width = 3.5in] {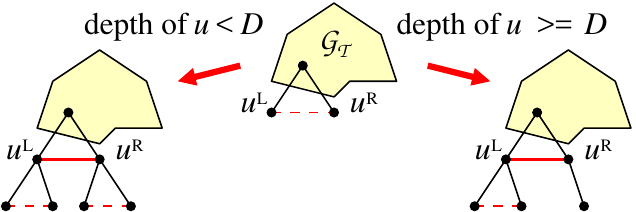}
\end{center}
\caption{Adversary moves. Upward edges from new vertices are not shown, to avoid clutter. Dashed
lines represent cross edges that are not collected by $\calS$, while thick lines represent
those that are already collected by $\calS$.}
\label{fig: adversary moves}
\end{figure}


Thus the adversary will be building the core binary skeleton tree down to depth $D$,
and from then on, if the game still continues, it will extend the tentacles. 
Our objective is to prove that, in each step, right after the adversary extends the graph but before 
$\calS$ updates its clustering, we have
\begin{equation}
	\profit_{\smallOPT}(\calT)  \ge (6-\epsilon_D) \cdot \profit_\calS(\calT),
	\label{eqn: 6-eps lower bound}
\end{equation}
where $\epsilon_D \to 0$ when $D\to\infty$.
This is enough to prove the lower bound of  $6-\epsilon_D$ on the absolute ratio. 
The reason is this:
If $\calS$ does not collect any edges at some step, the game stops, the ratio is 
 $6-\epsilon_D$, and we are done. Otherwise, the adversary will stop the game
after $2^{D+1} + M$ steps, where $M$ is some large integer. Then the profit of
$\calS$ is bounded by $2^{D+1} + M$ (the number of steps) plus the number
of remaining cross edges, and there are at most $2^D$ of those, so $\calS$'s profit
is at most $2^{D+2}+M$.
At that time, $\calT$ will have at least $M$ nodes in tentacles and at
most $2^D$ tentacles, so there is at least one tentacle of length $M/2^D$,
and this tentacle contributes $\Omega((M/2^D)^2)$ edges to the optimum.
Thus for $M$ large enough, the ratio between the optimal profit and the
profit of $\calS$ will be larger than $6$ (or any constant, in fact).

Once we establish~(\ref{eqn: 6-eps lower bound}), the lower bound of $6$ will
follow, because for any fixed $R < 6$ we can take $D$ large enough to
get a lower of $6-\epsilon_D \ge R$.


\paragraph{Computing the adversary's profit.}
We now explain how to estimate the adversary's profit for $\calG$. To this end, we
provide a specific recipe for computing a clique clustering of $\calG$. We do not
claim that this particular clustering is actually optimal, but it is a lower bound on 
the optimum profit, and thus it is sufficient for our purpose.

For any node $v\in\calT$ that is not a leaf, denote by $\calP^\LL(v)$ the longest path from $v$ to a
leaf of $\calT$ that goes through the left child of $v$.
If $v$ is a non-leaf in the core tree, and thus has a right child, then $\calP^\RR(v)$ is the longest path from $v$ to a
leaf of $\calT$ that goes through this right child. In both cases, ties are broken
arbitrarily but consistently, for example in favor of the leftmost leaves.
If $v$ is in a tentacle (so it does not have the right child), then we let
$\calP^\RR(v)=\braced{v}$.

Let $\calP^\LL(v) = (v=v_1,v_2,...,v_m)$, where $v_m$ is a leaf of $\calT$. 
Since $v$ is not a leaf, the definition of $\calT$ implies that $m\ge 2$.
We now define the clique $C^\LL(v)$ in $\calG$ that corresponds to $\calP^\LL(v)$. Intuitively,
for each $v_i$ we add to $C^\LL(v)$ one of the corresponding vertices,
$v_i^\LL$ or $v_i^\RR$, depending on whether $v_{i+1}$ is the left 
or to the right child of $v_i$.
The following formal definition describes the construction of $C^\LL(v)$ in a top-down fashion:
\begin{itemize}
	\item $v_1^\LL\in C^\LL(v)$. 
	\item Suppose that $1\le i \le m-1$
	 		and that $v_i^\sigma\in C^\LL(v)$, for $\sigma\in\braced{\rmL,\rmR}$. Then
	\begin{itemize}
		\item if $i = m-1$, add $v_m^\LL$ and $v_m^\RR$ to $C^\LL(v)$; 
		\item otherwise, if $v_{i+2}$ is the left child of $v_{i+1}$, add $v_{i+1}^\LL$ to $C^\LL(v)$,
				and if $v_{i+2}$ is the right child of $v_{i+1}$, add $v_{i+1}^\RR$ to $C^\LL(v)$.
	\end{itemize}
\end{itemize}
We define $C^\RR(v)$ analogously to $C^\LL(v)$, but with two differences.
One, we use $\calP^\RR(v)$ instead of $\calP^\LL(v)$. Two, if $v$ is in a tentacle then
we let $C^\RR(v) = \braced{v^\RR,v^\DD}$. In other words, the whiskers form 2-cliques.

Observe that except cliques $C^\RR(v)$ corresponding to the whiskers (that is, when $v$ is in a tentacle),
all cliques $C^\sigma(v)$ have cardinality at least~$3$.


\begin{figure}[ht]
\begin{center}
\includegraphics[width = 1.75in] {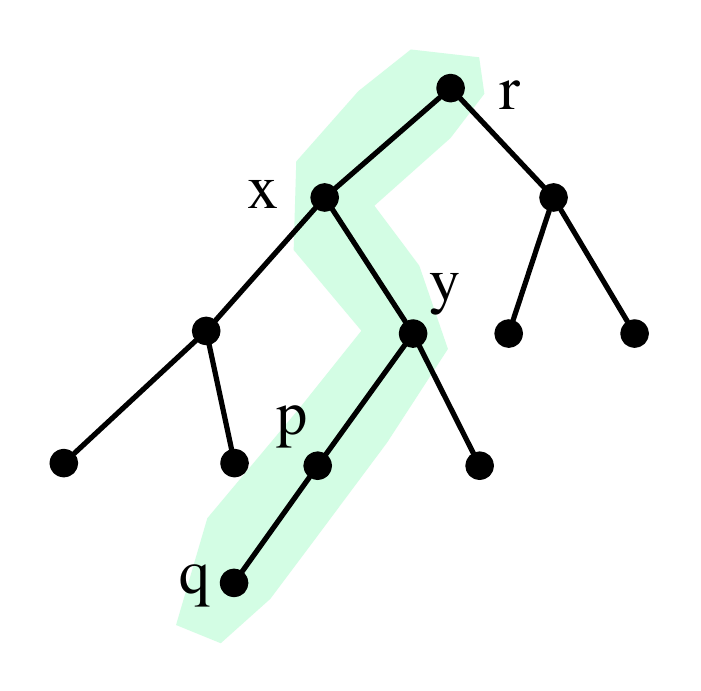}
\hskip 0.3in
\includegraphics[width = 3.9in] {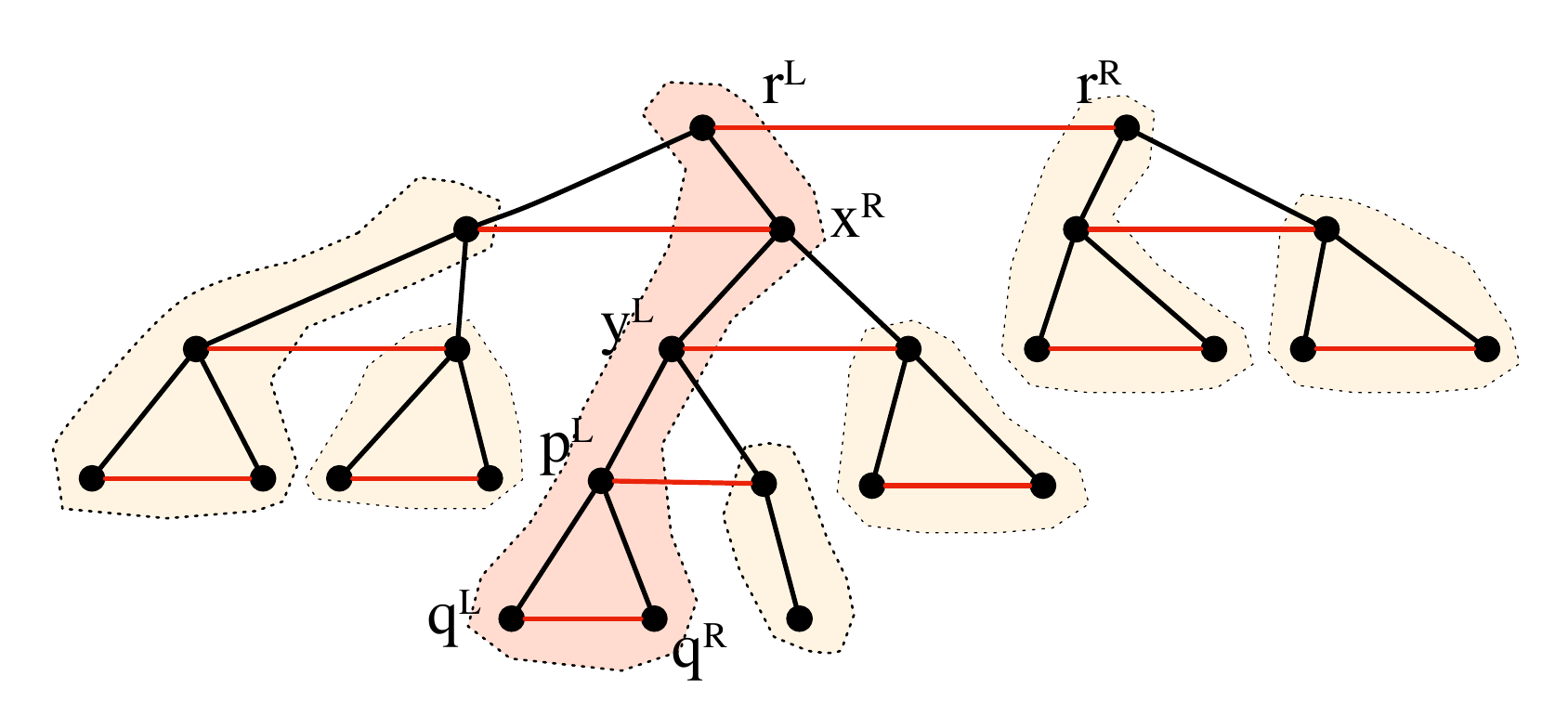}
\end{center}
\caption{On the left, an example of a path $\calP^\LL(r) = (r,x,y,p,q)$ in $\calT$. In this example, $D = 3$.
The corresponding clique $C^\LL(r)$ is shown on the right (darker shape). The figure on the right also shows
the adversary clique partitioning of $\calG$. To avoid clutter, upward edges are not shown.}
\label{fig: lower bound cliques}
\end{figure}


We now define a clique partitioning $\calC^\ast$ of $\calG$, as follows: First 
we include cliques $C^\LL(r)$ and $C^\RR(r)$ in $\calC^\ast$. 
We then proceed recursively: choose any node $v$ such that exactly one of $v^\LL, v^\RR$ is already covered
by some clique of $\calC^\ast$.
If $v^\LL$ is covered but $v^\RR$ is not, then include $C^\RR(v)$ in $\calC^\ast$.
Similarly, if  $v^\RR$ is covered but $v^\LL$ is not, then include $C^\LL(v)$ in~$\calC^\ast$.


\paragraph{Analysis.}
Denote by $\calT_v$ the subtree of $\calT$ rooted at $v$. By $\calG_v$ we denote the
subgraph of $\calG$ induced by the vertices that correspond to the nodes in $\calT_v$.
Each clique in $\calC^\ast$ that intersects $\calG_v$ induces a clique in $\calG_v$, 
and the partitioning $\calC^\ast$
induces a partitioning $\calC^\ast_v$ of $\calG_v$ into cliques.
We will use notation $\optprof_v$ for the profit of partitioning $\calC^\ast_v$.
Note that $\calC^\ast_v$ can be obtained with the same top-down process as $\calC^\ast$, but
starting from $v$ as the root instead of~$r$.

We denote strategy $\calS$'s profit (the number of cross edges) within $\calG_v$ by $\algprof_v$ .
In particular, we have $\profit_\calS(\calG) = \algprof_r$
and $\profit_{\smallOPT}(\calG) \ge \optprof_r$. 
Thus, to show~(\ref{eqn: 6-eps lower bound}), it is sufficient to prove that
\begin{equation}
	 \optprof_r \ge (6-\epsilon_D) \cdot \algprof_r,
	\label{eqn: 6-eps lower bound 2}
\end{equation}
where $\epsilon_D\to 0$ when $D\to\infty$.

We will in fact prove an analogue of inequality~(\ref{eqn: 6-eps lower bound 2})  for all
subtrees $\calT_v$. To this end, we distinguish between two types of subtrees $\calT_v$.
If $\calT_v$ ends at depth $D$ of $\calT$ or less (in other words, if $\calT_v$ is inside the core of $\calT$), 
we call $\calT_v$ \emph{shallow}.
If $\calT_v$ ends at depth $D+1$ or more, we call it \emph{deep}. So deep subtrees are those that
contain some tentacles of~$\calT$.


\begin{lemma}\label{lem: lb 6 shallow}
If $\calT_v$ is shallow, then
\begin{equation*}
	\optprof_v \ge 6 \cdot \algprof_v.	
\end{equation*}
\end{lemma}

\begin{proof}
This can be shown by induction on the depth of $\calT_v$. If this depth is $0$, that is
$\calT_v = \braced{v}$, then $\optprof_v = 1$ and $\algprof_v = 0$, so the ratio is
actually infinite. To jump-start the induction we also need 
to analyze the case when the depth of $\calT_v$ is $1$. This means that
$\calS$ collected only edge $(v^\LL,v^\RR)$ from $\calT_v$. When this happened,
the adversary generated vertices corresponding to the two children of $v$ in $\calT$
and his clustering will consist of two triangles. So now $\optprof_v = 6$ and $\algprof_v=1$, 
and the lemma holds.

Inductively, suppose that the depth of $\calT_v$ is at least two, let $y,z$ be the
left and right children of $v$ in $\calT$, and assume that
the lemma holds for $\calT_y$ and $\calT_z$.
Naturally, we have $\algprof_v = \algprof_y + \algprof_z + 1$.
Regarding the adversary profit, since the depth of $\calT_v$ is at least two, cluster
$C^\LL(v)$ contains exactly one of $y^\LL, y^\RR$; say it contains $y^\LL$.
Thus $C^\LL(v)$ is obtained from $C^\LL(y)$ by adding $v^\LL$.
By the definition of clustering $\calC^\ast$, the depth of $\calT_y$ is at least $1$,
which means that adding $v^\LL$ will add at least three new edges.
By a similar argument, we will also add at least three edges from $v^\RR$. This implies that
$\optprof_v\ge  \optprof_y + \optprof_z + 6 
		\ge 6\cdot\algprof_y + 6\cdot\algprof_z + 6 
		\ge 6\cdot\algprof_v$,
completing the inductive step.
\end{proof}

From Lemma~(\ref{lem: lb 6 shallow}) we obtain that, in particular, if $\calT$ itself is shallow then
$\optprof_r\ge 6 \cdot \algprof_r$, which is even stronger than
inequality~(\ref{eqn: 6-eps lower bound 2}) that we are in the process of justifying.
Thus, for the rest of the proof, we can restrict our attention to skeleton 
trees $\calT$ that are deep.

So next we consider deep subtrees of $\calT$.
The \emph{core depth} of a deep subtree $\calT_v$ is defined as the depth of
the part of $\calT_v$ within the core subtree of $\calT$. (In other words, the
core depth of $\calT_v$ is equal to $D$ minus the depth of $v$ in $\calT$.)
If $h$ and $s$ are, respectively, the core depth of $\calT_v$ and its
maximum tentacle length, then $0\le h\le D$ and $s\ge 1$.
The sum $h+s$ is then simply the depth of $\calT$.


\begin{lemma}\label{lem: bottom subtree ratio}
Let $\calT_v$ be a deep subtree of core depth $h\ge 0$ and maximum tentacle length $s\ge 1$, then
\begin{equation*}
		\optprof_v + 2(h+s) \ge 6\cdot \algprof_v.
\end{equation*}
\end{lemma}

Before proving the lemma, let us argue first that this lemma is sufficient
to establish our lower bound. Indeed, since we are now considering the case when 
$\calT$ is a deep subtree itself, the lemma implies that
$\optprof_r + 2(D+s) \ge 6\cdot \algprof_r$, where $s$ is the maximum tentacle length of $\calT$.
But $\optprof_r$ is at least quadratic in $D+s$. So for large $D$ the
ratio $\optprof_r/\algprof_r$ approaches $6$.

\medskip

\begin{proof}
To prove Lemma~\ref{lem: bottom subtree ratio}, we use
induction on $h$, the core depth of $\calT_v$. Consider first the base case, for $h=0$ (when $\calT_v$ is just a tentacle).
In his clustering  $\calC^\ast_v$,
the adversary has one clique of $s+2$ vertices, namely all $x^\LL$ vertices
in the tentacle (there are $s+1$ of these), plus one $z^\RR$ vertex for the leaf $z$.
He also has $s$ whiskers, so his profit for $\calT_v$ is
$\binom{s+2}{2} + s = \half(s^2+5s+2)$. 
$\calS$ collects only $s$ edges, namely all cross edges in $\calT_v$ except the last.
(See Figure~\ref{fig: inductive proof base}.)
Solving the quadratic inequality and using the integrality of $s$, we get 
$\optprof_v + 2s \ge 6s = 6\cdot\algprof_v$. Note that this inequality is in fact tight for $s = 1$ and~$2$.

\begin{figure}
\begin{center}
\includegraphics[width = 2.2in]{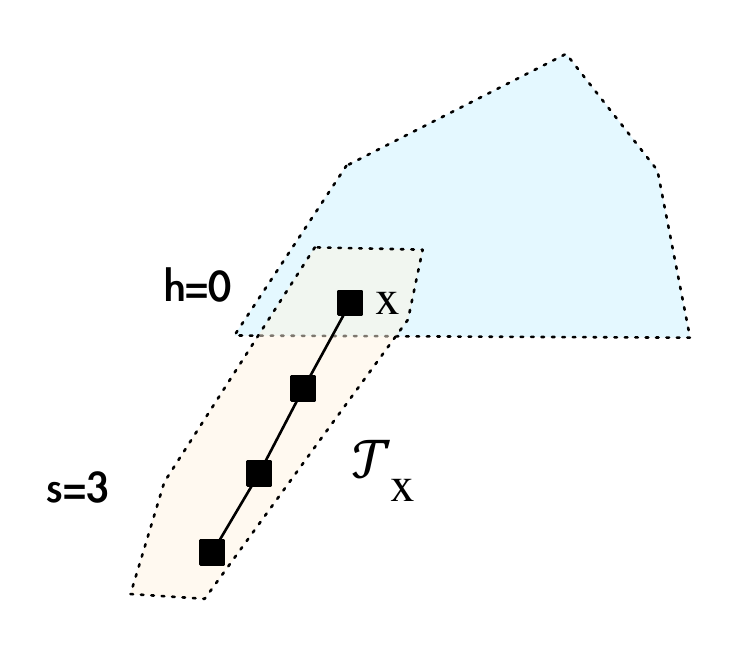}
\hskip 0.5in
\includegraphics[width = 2.2in]{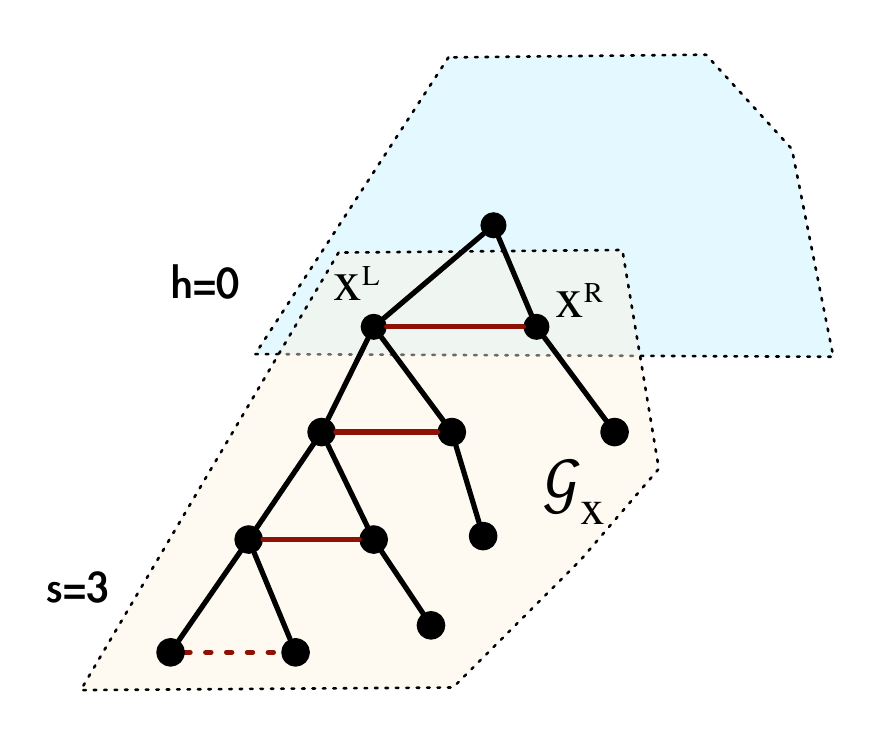}
\end{center}
\caption{Illustration of the proof of Lemma~\ref{lem: bottom subtree ratio}, the base case. 
Subtree $\calT_v$ on the left, the corresponding subgraph $\calG_v$ on the right.}
\label{fig: inductive proof base}
\end{figure}

In the inductive step, consider a deep subtree $\calT_v$. Let
$y$ and $z$ be the left and right children of $v$. Without loss of
generality, we can assume that 
$\calT_y$ is a deep tree with core depth $h-1$ and the same maximum
tentacle length $s$ as $\calT_v$, while $\calT_z$ is either shallow
(that is, it has no tentacles), or it is a deep tree with maximum
tentacle length at most $s$.

\begin{figure}
\begin{center}
	\includegraphics[width = 2in]{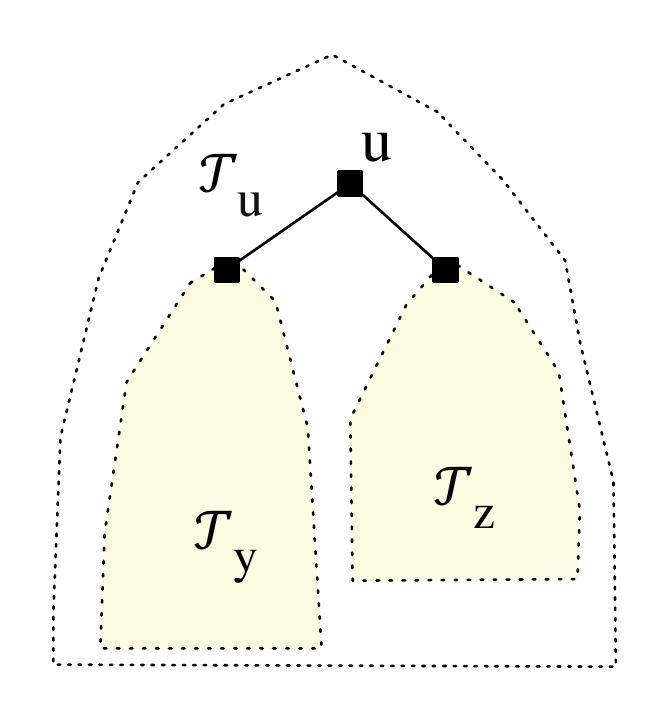}
	\hskip 0.5in
	\includegraphics[width = 2in]{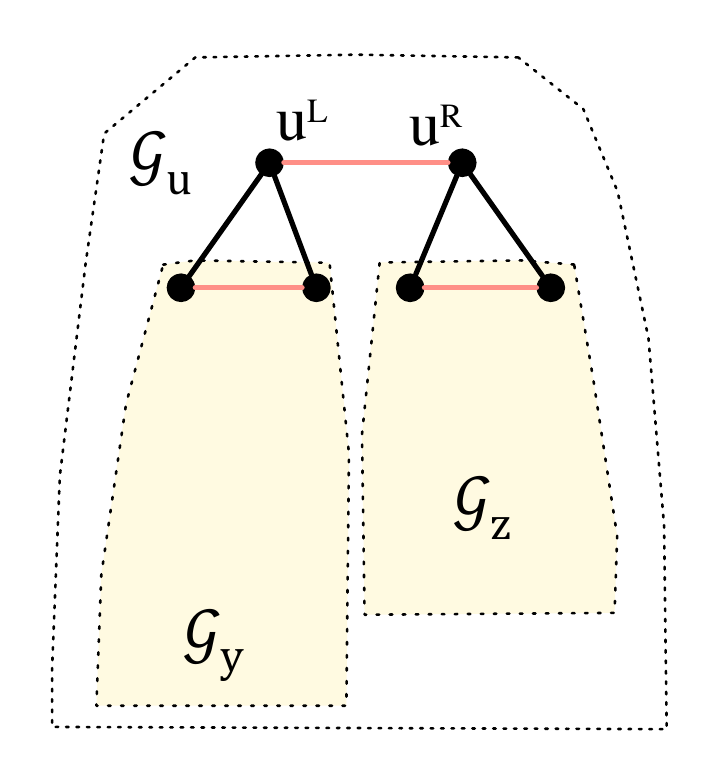}
\end{center}
\caption{Illustration of the proof of Lemma~\ref{lem: bottom subtree ratio}, the inductive step. 
Subtrees $\calT_v,\calT_y,\calT_z$ on the left, the corresponding subgraphs on the right.}
\label{fig: inductive proof ind step}
\end{figure}

By the inductive assumption, we have $\optprof_y + 2(h-1+s) \ge 6\cdot \algprof_y$.
Regarding $z$, if $\calT_z$ is shallow then from Lemma~\ref{lem: lb 6 shallow}
we get $\optprof_z  \ge 6\cdot \algprof_z$, and if $\calT_z$ is deep (necessarily of core depth $h-1$) then
$\optprof_z + 2(h-1+s') \ge 6\cdot \algprof_z$, where $s'$ is $\calT_z$'s maximum tentacle length, such that $1\le s' \le s$.

Consider first the case when $\calT_z$ is shallow. 
Note that
\begin{align*}
	\algprof_v &= \algprof_y + \algprof_z + 1
	\ \ \mbox{and} 
	\\
    \optprof_v &\ge \optprof_y + \optprof_z + h+s+4
\end{align*}
The first equation is trivial, because the profit of $\calS$ in $\calG_v$ consists of all cross edges
in $\calG_y$ and $\calG_z$, plus one more cross edge $(v^\LL,v^\RR)$. The second inequality
holds because the adversary clustering $\calC^\ast_v$ is obtained by adding
$v^\LL$ to
$\calG_y$'s cluster with $(h-1) +s+2 = h+s+1$ vertices, and $v^\RR$ can be added to
$\calG_z$'s cluster that with at least $3$ vertices.
We get
\begin{align*}
	\optprof_v + 2(h+s) &\ge [\optprof_y + \optprof_z + h+s+4] + 
										2(h+s)
						\\				
					&\ge [\optprof_y  + 2(h-1+s)   ] + \optprof_z  + 6
						\\
					&\ge 6\cdot \algprof_y	+ 6\cdot \algprof_z
							+ 6
						\\
					&= 6\cdot \algprof_v.
\end{align*}

The second case is when $\calT_z$ is a deep tree (of the same core depth $h-1$ as $\calT_y$)
with maximum tentacle length $s'$, where $1\le s'\le s$. 
As before, we have $\algprof_v = \algprof_y + \algprof_z + 1$. The optimum
profit satisfies (by a similar argument as before, applied to both $\calT_y$ and $\calT_z$)
\begin{align*}
	\optprof_v &\ge \optprof_y + \optprof_z + 2h+s+s'+2.
\end{align*}
We obtain (using $s\ge s'$)
\begin{align*}
	\optprof_v + 2(h+s) &\ge [\optprof_y + \optprof_z + 2h+s+s'+2] + 
										2(h+s)
		\\
		&\ge [\optprof_y + 2(h-1+s)] + [\optprof_z + 2(h-1+s')  ] + 6
		\\
		&\ge 6\cdot\algprof_y + 6\cdot\algprof_z + 6
		\\
		&= 6\cdot\algprof_v.
\end{align*}
This completes the proof of Lemma~\ref{lem: bottom subtree ratio}.
\end{proof}

\smallskip

We still need to explain how to extend our proof so that it also applies to
the asymptotic competitive ratio. This is quite simple: Choose some large
constant $K$. The adversary will create $K$ instances of the above game,
playing each one independently. Our construction above uses the fact that at
each step the strategy is forced to collect one of the pending cross
edges, otherwise its competitive ratio exceeds ratio $R$ (where
$R$ is arbitrarily close to $6$). Now, for $K$ sufficiently large, the strategy is
forced to collect cross edges in all except for some finite number of
copies of the game, where this number depends on the additive constant in the
competitiveness bound. 

\smallskip

\emph{Note:} Our construction is very tight, in the following sense. Suppose that
$\calS$ maintains $\calT$ as balanced as possible. Then
the ratio is exactly $6$ when the depth of $\calT$ is $1$ or $2$.
Furthermore, suppose that $D$ is very large and the strategy constructs 
$\calT$ to have depth $D$ or more, that is, it starts growing tentacles (but still
maintaining $\calT$ balanced.)
Then the ratio is $6-o(1)$ for tentacle lengths $s=1$ and $s=2$.
The intuition is that when the adversary plays optimally, he will only
allow the online strategy to collect isolated edges (cliques of size $2$).
For this reason, we conjecture that $6$ is the optimal competitive ratio.


\section{Online \MinCC\ Clustering}
\label{sec: minclustering}


In this section, we study the clique clustering problem with a different measure of optimality
that we call $\MinCC$. For $\MinCC$, we define the \emph{cost} of a clustering $\calC$ to be the
total number of {\em non-cluster edges}. Specifically, if the cliques in $\calC$ are
$C_1,C_2,...,C_k$ then the cost of $\calC$ is
$|E|-\sum_{i=1}^k\binom{|C_i|}{2}$. The objective is to construct
a clustering that minimizes this cost.


\subsection{A Lower Bound for Online {\MinCC} Clustering}\label{subsec: min lower bound}

In this section we present a lower bound for deterministic {\MinCC} clustering.

\begin{theorem}\label{thm:minlb}
\textrm{(a)}
There is no online strategy for $\MinCC$ clustering with competitive ratio $n - \omega(1)$,
where $n$ is the number of vertices.

\textrm{(b)}
There is no online strategy for $\MinCC$ clustering with absolute
competitive ratio smaller than~$n - 2$.
\end{theorem}

\begin{proof}
(a)
Consider a strategy $\calS$ with competitive ratio $R_n = n-\omega(1)$.
Thus, according to the definition~(\ref{eqn: min competitive ratio}) of the competitive
ratio, there is a constant $\beta$ that satisfies
$\cost_{\calS}(G) \le R_n \cdot\cost_\smallOPT(G) + \beta$, where $n = |G|$.
We can assume that $\beta$ is a positive integer.

\begin{figure}
\begin{center}
\includegraphics{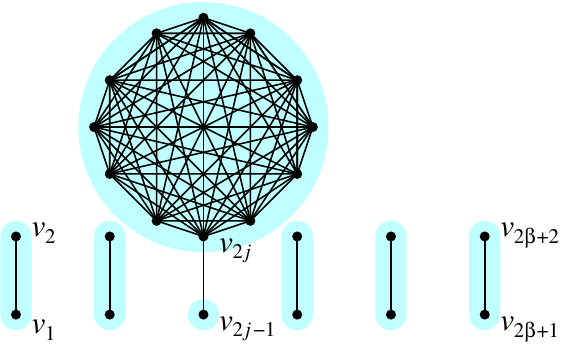}
\caption{\label{fig: mincc lower bound}Illustrating the lower bound proof of Theorem~\protect\ref{thm:minlb}. The figure shows the optimal clustering for the graph.}
\end{center}
\end{figure}

The adversary first produces a graph of
$2\beta+2$ vertices connected by $\beta+1$ disjoint edges $(v_{2i-1},v_{2i})$,
for $i = 1,2,\ldots,\beta+1$. At this point,
$\calS$ must have added at least one pair $\braced{v_{2j-1},v_{2j}}$ to its clustering, because
otherwise, since $\cost_\smallOPT(G)=0$, inequality~(\ref{eqn: min competitive ratio}) would be
violated.
The adversary then chooses some large $n$ and adds $n-2\beta-2$
new vertices $v_{2\beta+3},\ldots,v_{n}$ that together with $v_{2j}$ form a clique of size
$n - 2\beta-1$; see Figure~\ref{fig: mincc lower bound}.
All edges from $v_{2j}$ to these new vertices are non-cluster edges for $\calS$
and the optimum solution has only one non-cluster edge $(v_{2j-1},v_{2j})$.
Thus
\begin{equation*}
\cost_{\calS}(G) - \beta \ge (n-2\beta-2) - \beta
						= n-3\beta - 2
						= (n-3\beta-2) \cdot \cost_\smallOPT(G)
						> R_n\cdot\cost_\smallOPT(G),
\end{equation*}
if $n$ is large enough, giving us a contradiction.

\smallskip
(b) The proof of this part is a straightforward modification of the
proof for (a): the adversary starts by releasing just one edge $(v_1,v_2)$,
and the online strategy is forced to cluster $v_1$ and $v_2$ together,
because now $\beta = 0$. Then the adversary forms a clique of size
$n-1$ including $v_2$. The details are left to the reader.
\end{proof}


\subsection{The Greedy Strategy for Online {\MinCC} Clustering}\label{subsec:mingreedy}

We continue the study of online {\MinCC} clustering, and we prove that {\Greedy},
the greedy strategy presented in Section~\ref{subsec: greedy_algorithm},
yields a competitive ratio matching the lower bound from the previous section.


\begin{theorem}\label{thm:mincc-upper}
The absolute competitive ratio of {\Greedy} is $n-2$.
\end{theorem}

\begin{proof}
The key observation for this proof is that, for any triplet of vertices $u$, $v$, and $v'$,
if the graph contains the two edges $(u,v)$ and $(u,v')$ but $v$ and $v'$ are not connected by an edge,
then in any clustering at least one of the edges $(u,v)$ or $(u,v')$ is a non-cluster edge.


\emmedparagraph{Claim~A:}
Let $(u,v)$ be a non-cluster edge of {\Greedy}. Then {\OPT}
(the optimal clustering) has at least one non-cluster
edge adjacent to $u$ or $v$ (which might also be $(u,v)$ itself).

\medskip
Without loss of generality suppose vertex $v$ arrives after vertex $u$.
Let $A$ be the cluster of {\Greedy} containing vertex $u$ at the moment when
vertex $v$ arrives. We have that $v\notin A$.
If $A$ contains some vertex $u'$ not connected to $v$,
then the earlier key observation shows that one
of the edges $(u',u)$, $(u,v)$ is a non-cluster edge for {\OPT}; see Figure~\ref{fig:minCC}.

\begin{figure}
    \centerline{\includegraphics[width=8cm]{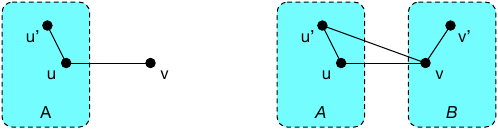}}
    \caption{Illustration of the proof of Theorem~\ref{thm:mincc-upper}.}
    \label{fig:minCC}
\end{figure}

Now assume that $v$ is connected to all vertices of $A$. {\Greedy} had an option of
adding $v$ to $A$ and it didn't, so it placed $v$ in some clique $B$ (of size at least $2$)
that is not merge-able with $A$, that is,
there are vertices $u'\in A$ and $v'\in B$ which are not connected by an edge.
Now the earlier key observation shows that
one of the edges $(u',v)$, $(v,v')$ is a non-cluster edge of {\OPT}. This completes
the proof of Claim~A.

\smallskip

To estimate the number of non-cluster edges of {\Greedy}, we use a charging scheme.
Let $(u,v)$ be a non-cluster edge of {\Greedy}. We charge it to non-cluster edges
of {\OPT} as follows.
\begin{description}
\item{\emph{Self charge:}}
If $(u,v)$ is a non-cluster edge of {\OPT}, we charge $1$ to $(u,v)$ itself.
\item{\emph{Proximate charge:}}
If $(u,v)$ is a cluster edge in {\OPT}, we split the charge of $1$ from $(u,v)$
evenly among all non-cluster edges of {\OPT} incident to $u$ or $v$.
\end{description}
From Claim~A, the charging scheme is well-defined, that is,
all non-cluster edges of {\Greedy} have been charged fully to non-cluster
edges of {\OPT}.
It remains to estimate the total charge that any non-cluster edge of {\OPT} may have received.
Since the absolute competitive ratio is the ratio between the number of non-cluster edges
of {\Greedy} and the number of non-cluster edges of {\OPT}, the maximum charge to any
non-cluster edge of {\OPT} is an upper bound for the absolute competitive ratio.

Consider a non-cluster edge $(x,y)$ of {\OPT}. Edge $(x,y)$ can receive charges only from
itself (self charge) and other edges incident to $x$ or $y$ (proximate charges). Let
$P$ be the set of vertices adjacent to both $x$ and $y$, and let
$Q$ be the set of vertices that are adjacent to only one of them, but excluding $x$ and $y$:
\begin{equation*}
	P = N(x) \cap N(y)
		\quad\textrm{and}\quad
	Q = N(x)\cup N(y) - P  - \braced{x,y}.
\end{equation*}
($N(z)$ denotes the neighborhood of a vertex $z$, the set of vertices adjacent to $z$.)
We have $|P|+|Q|\leq n-2$.

Edges connecting $x$ or $y$ to $Q$ will be called $Q$-edges. Trivially,
the total charge from $Q$-edges to $(x,y)$ is at most $|Q|$.

Edges connecting $x$ or $y$ to $P$ will be called $P$-edges. Consider
some $z\in P$. Since $x$ and $y$ are in different clusters of {\OPT},
at least one of $P$-edges $(x,z)$ or $(y,z)$ must also be a
non-cluster edge for {\OPT}. By symmetry, assume that $(x,z)$ is a non-cluster
edge for {\OPT}. If $(x,z)$ is a non-cluster edge of {\Greedy} then
$(x,z)$ will absorb its self charge. So $(x,z)$ will not contribute to
the charge of $(x,y)$. If $(y,z)$ is a non-cluster edge of {\Greedy}
then either it will be self charged (if it's also a non-cluster edge of {\OPT})
or its proximate charge will be split between at least two edges,
namely $(x,y)$ and $(x,z)$. Thus the charge from
$(y,z)$ to $(x,y)$ will be at most $\half$.
Therefore the total charge from $P$-edges to $(x,y)$ is at most $\half |P|$.
We now have some cases.


\smallskip
\noindent
\mycase{1}
$(x,y)$ is a cluster edge of {\Greedy}. Then $(x,y)$ does not generate a self charge,
so the total charge received by $(x,y)$ is at most
$\half |P| + |Q| \le |P| + |Q| \le n-2$.


\smallskip
\noindent
\mycase{2}
$(x,y)$ is a non-cluster edge of {\Greedy}. Then $(x,y)$ contributes a
self charge to itself.

\begin{description}

	\item{\mycase{2.1}}
 		$|P|\ge 2$. Then $\half |P|\le |P|-1$, so the total charge received
		by $(x,y)$ is at most
		$\half |P| + |Q| +1 \le (|P|-1) +|Q|
						 	= |P|+|Q|
							\le n-2$.

	\item{\mycase{2.2}}
		At least one $Q$-edge is a cluster edge of {\Greedy}.
		Then the total proximate charge from $Q$-edges is
		at most $|Q|-1$, so the total
		charge received by $(x,y)$ is at most
		$\half |P| + (|Q|-1) + 1 \le  |P| + |Q|
								\le n-2$.

	\item{\mycase{2.3}}
		$|P| \in\braced{0,1}$ and all $Q$-edges are non-cluster edges of {\Greedy}.
		We claim that this case cannot actually occur. Indeed,
		if $|P| = 0$ then {\Greedy} would cluster $x$ and $y$ together.
		Similarly, if $P = \braced{z}$,
		then {\Greedy} would cluster $x$, $y$ and $z$ together.
		In both cases, we get a contradiction with the assumption of Case~2.
\end{description}

Summarizing, we have shown that each non-cluster edge of
{\OPT} receives a total charge of at most $n-2$, and
the theorem follows.
\end{proof}

The proof of Theorem~\ref{thm:mincc-upper} applies in fact to a more general
class of strategies, giving an
upper bound of $n-2$ on the absolute competitive ratio of
all ``non-procrastinating'' strategies, which
never leave merge-able clusters in their clusterings
(that is clusters $C$, $C'$ such that $C\cup C'$ forms a clique).


\bibliographystyle{plain}
\bibliography{clique_clustering}

\begin{thebibliography}{10}

\bibitem{BansalBC04_correlation_2004}
Nikhil Bansal, Avrim Blum, and Shuchi Chawla.
\newblock Correlation clustering.
\newblock {\em Machine Learning}, 56(1-3):89--113, 2004.

\bibitem{Ben-DorSY_gene_expression_1999}
Amir Ben{-}Dor, Ron Shamir, and Zohar Yakhini.
\newblock Clustering gene expression patterns.
\newblock {\em Journal of Computational Biology}, 6(3/4):281--297, 1999.

\bibitem{Borodin_ElYaniv_98}
Allan Borodin and Ran El{-}Yaniv.
\newblock {\em Online computation and competitive analysis}.
\newblock Cambridge University Press, 1998.

\bibitem{CharikarCFM_incremental_2004}
Moses Charikar, Chandra Chekuri, Tom{\'{a}}s Feder, and Rajeev Motwani.
\newblock Incremental clustering and dynamic information retrieval.
\newblock {\em {SIAM} J. Comput.}, 33(6):1417--1440, 2004.

\bibitem{charikar2003clustering}
Moses Charikar, Venkatesan Guruswami, and Anthony Wirth.
\newblock Clustering with qualitative information.
\newblock In {\em Foundations of Computer Science, 2003. Proceedings. 44th
  Annual IEEE Symposium on}, pages 524--533. IEEE, 2003.

\bibitem{Chaudhuri_GRT_tours_03}
Kamalika Chaudhuri, Brighten Godfrey, Satish Rao, and Kunal Talwar.
\newblock Paths, trees, and minimum latency tours.
\newblock In {\em 44th Symposium on Foundations of Computer Science {(FOCS}
  2003), 11-14 October 2003, Cambridge, MA, USA, Proceedings}, pages 36--45,
  2003.

\bibitem{Chrobak_etal_15}
Marek Chrobak, Christoph D\"urr, and Bengt~J. Nilsson.
\newblock Competitive strategies for online clique clustering.
\newblock In {\em Proc. 9th International Conference on Algorithms and
  Complexity ({CIAC}'15)}, pages 101--113, 2015.

\bibitem{Chrobak_Hurand_medians_11}
Marek Chrobak and Mathilde Hurand.
\newblock Better bounds for incremental medians.
\newblock {\em Theor. Comput. Sci.}, 412(7):594--601, 2011.

\bibitem{Chrobak_KNY_bidding_08}
Marek Chrobak, Claire Kenyon, John Noga, and Neal~E. Young.
\newblock Incremental medians via online bidding.
\newblock {\em Algorithmica}, 50(4):455--478, 2008.

\bibitem{Chrobak_Mathieu_doubling_2006}
Marek Chrobak and Claire Kenyon{-}Mathieu.
\newblock {SIGACT} news online algorithms column 10: competitiveness via
  doubling.
\newblock {\em {SIGACT} News}, 37(4):115--126, 2006.

\bibitem{DemaineImorlica03}
Erik~D. Demaine and Nicole Immorlica.
\newblock Correlation clustering with partial information.
\newblock In {\em Proc. 6th International Workshop on Approximation Algorithms
  for Combinatorial Optimization Problems (APPROX'03)}, pages 1--13, 2003.

\bibitem{Dessmark_etal_06}
Anders Dessmark, Jesper Jansson, Andrzej Lingas, Eva-Marta Lundell, and Mia
  Persson.
\newblock On the approximability of maximum and minimum edge clique partition
  problems.
\newblock In {\em Proceedings of the 12th Computing: The Australasian Theory
  Symposium (CATS'06)}, pages 101--105, 2006.

\bibitem{FabijanNP_clique_clustering_2013}
Aleksander Fabijan, Bengt~J. Nilsson, and Mia Persson.
\newblock Competitive online clique clustering.
\newblock In {\em Proc. 8th International Conference on Algorithms and
  Complexity ({CIAC}'13)}, pages 221--233, 2013.

\bibitem{FigueroaBJ_fingerprings_2004}
Andres Figueroa, James Borneman, and Tao Jiang.
\newblock Clustering binary fingerprint vectors with missing values for {DNA}
  array data analysis.
\newblock {\em Journal of Computational Biology}, 11(5):887--901, 2004.

\bibitem{FigueroaGJKLP_approximate_clustering_2008}
Andres Figueroa, Avraham Goldstein, Tao Jiang, Maciej Kurowski, Andrzej Lingas,
  and Mia Persson.
\newblock Approximate clustering of incomplete fingerprints.
\newblock {\em J. Discrete Algorithms}, 6(1):103--108, 2008.

\bibitem{Lin_NRW_incremental_10}
Guolong Lin, Chandrashekhar Nagarajan, Rajmohan Rajaraman, and David~P.
  Williamson.
\newblock A general approach for incremental approximation and hierarchical
  clustering.
\newblock {\em {SIAM} J. Comput.}, 39(8):3633--3669, 2010.

\bibitem{MathieuSS_correlation_2010}
Claire Mathieu, Ocan Sankur, and Warren Schudy.
\newblock Online correlation clustering.
\newblock In {\em 27th International Symposium on Theoretical Aspects of
  Computer Science ({STACS}'10)}, pages 573--584, 2010.

\bibitem{ShamirST_graph_modification_2004}
Ron Shamir, Roded Sharan, and Dekel Tsur.
\newblock Cluster graph modification problems.
\newblock {\em Discrete Applied Mathematics}, 144(1-2):173--182, 2004.

\bibitem{Valinsky_etal_bacterial_2002}
Lea Valinsky, Gianluca~Della Vedova, Ra~J. Scupham, Sam Alvey, Andres Figueroa,
  Bei Yin, R.~Jack Hartin, Marek Chrobak, David~E. Crowley, Tao Jiang, and
  James Borneman.
\newblock Analysis of bacterial community composition by oligonucleotide
  fingerprinting of {rRNA} genes.
\newblock {\em Applied and Environmental Microbiology}, 68:2002, 2002.

\end{thebibliography}

\end{document}